\documentclass[10pt]{article}

\usepackage{authblk}

\usepackage{graphics} % for pdf, bitmapped graphics files
\usepackage{epsfig} % for postscript graphics files
\usepackage{amsmath} % assumes amsmath package installed
\usepackage{amssymb}  % assumes amsmath package installed
\usepackage{amsthm}	
\usepackage{bm} % bold font mathematical symbols
\usepackage{cite}

\usepackage[a4paper, total={6in, 8in}]{geometry}

\usepackage[caption=false]{subfig}

\usepackage{xcolor}
\usepackage{color}
\usepackage{capt-of}
\usepackage{mathtools}
\usepackage{enumitem}

\usepackage{booktabs}
\usepackage{multirow}

\usepackage{footnote}
\makesavenoteenv{tabular}
\makesavenoteenv{table}

\usepackage{scrextend}
\usepackage{calc}
\newlist{Description}{description}{3}
\setlist[Description]{style=nextline,font =\normalfont}
\SetEnumitemKey{margin}{leftmargin={\widthof{#1}+1em}}

\newcommand{\R}{\mathbb{R}}
\newcommand{\Rs}{\mathcal{R}}

\newcommand{\itset}[1]{\mathsf{#1}}

\newcommand{\event}{\lozenge}
\newcommand{\always}{\square}

\renewcommand{\next}{\bigcirc}

\newcommand{\true}{\text{true}}

\newcommand{\rmm}[1]{\left< \mathrm{#1} \right>}
\newcommand{\rmms}[2]{\left< \mathrm{#1}_{#2} \right>}

\newtheorem{rem}{Remark}
\newtheorem{theorem}{Theorem}

\newtheorem{exmp}{Example}[section]
\newtheorem{defn}{Definition}[section]
\newtheorem{prob}{Problem}[section]

\DeclareMathOperator*{\argsup}{arg\,sup}
\DeclareMathOperator*{\arginf}{arg\,inf}

\newcommand\Tstrut{\rule{0pt}{2.6ex}}         % = `top' strut
\newcommand\Bstrut{\rule[-0.9ex]{0pt}{0pt}}   % = `bottom' strut

\newcommand{\rev}[1]{#1}

\begin{document}

{\centering{\huge Formal synthesis of closed-form sampled-data controllers for nonlinear continuous-time systems under STL specifications.\par}\vspace{3ex}
	{\large Cees F. Verdier$^a$, Niklas Kochdumper$^b$, Matthias Althoff$^b$, and Manuel Mazo Jr.$^a$  \par}\vspace{2ex}
	
	$~^a$ Delft Center for Systems and Control, Delft University of Technology, The Netherlands (e-mail: c.f.verdier@tudelft.nl)\par\vspace{4ex}
	$~^b$ Department of Informatics, Technical University of Munich, 85748 Garching, Germany \par\vspace{4ex}
	Supported by NWO Domain TTW under the CADUSY project \#13852 and the ERC Starting Grant SENTIENT (755953). \par\vspace{4ex}}
{\centering\bfseries Abstract\par}
\smallbreak
We propose a counterexample-guided inductive synthesis framework for the formal synthesis of closed-form sampled-data controllers for nonlinear systems to meet STL specifications \rev{over finite-time trajectories}. Rather than stating the STL specification for a single initial condition, we consider an (infinite and bounded) set of initial conditions. Candidate solutions are proposed using genetic programming, which evolves controllers based on a finite number of simulations. Subsequently, the best candidate is verified using reachability analysis; if the candidate solution does not satisfy the specification, an initial condition violating the specification is extracted as a counterexample. Based on this counterexample, candidate solutions are refined until eventually a solution is found \rev{(or a user-specified number of iterations is met)}. The resulting sampled-data controller is expressed as a closed-form expression, enabling \rev{both interpretability and} the implementation in embedded hardware with limited memory and computation power. The effectiveness of our approach is demonstrated for multiple systems.
\medbreak
\par\vspace{2ex}

\section{Introduction}
Recent years have seen a surge in interest in controller synthesis for temporal logic specifications, realizing complex behavior beyond traditional stability requirements, see, e.g., the recent literature survey in \cite{Belta2019}. Originally stemming from the field of computer science, temporal logic has been used to describe the correctness of complex behaviors of computer systems \cite{Baier2008}. As it originally dealt with finite systems, (bi-)simulation approaches have been proposed to abstract infinite systems to finite systems \cite{Tabuada2009,Belta2017}. However, as a downside, these approaches (e.g., \cite{Reissig2017,Liu2013, Habets2006, Girard2010}) typically suffer from the curse of dimensionality and return controllers in the form of enormous lookup tables \cite{Zapreev2018}. 

Where certain temporal logics such as linear temporal logic reason over traces of finite systems, signal temporal logic (STL) reasons over continuous signals \cite{maler2004}. Besides a Boolean answer to whether the formula is satisfied, quantitative semantics of STL has been introduced \cite{Donze2010,Fainekos2009}, providing a quantititive measure on how robustly a formula is satisfied. These robustness measures enable optimization-based methods for temporal logic, such as model predictive control (MPC) \cite{Raman2015,Sadraddini2015,Farahani2015,Sadraddini2018,Lindemann2017}, optimal trajectory planning \cite{Pant2018}, reinforcement learning \cite{Aksaray2016}, and neural networks \cite{Yaghoubi2019,Liu2021}. Apart from optimization-based methods, other proposed approaches for STL specifications rely on control barrier functions (CBF) \cite{Lindemann2019,Garg2019}. While the work in \cite{Lindemann2019} does not optimize a robustness measure of the STL specification, the computation of the control input for every time step relies on online quadratic optimization. Alternatively, in \cite{Lindemann2017b,Lindemann2019C} the synthesis for a fragment of STL is reformulated to a prescribed performance control problem, resulting in a continuous state feedback control law.

While (bi-)simulation approaches provide feedback strategies for all (admissible) initial conditions, only a limited number of optimization-based approaches consider a set of initial conditions \cite{Belta2019}, including \cite{Raman2015,Farahani2015,Schurmann2018}. In \cite{Sadraddini2018}, tube MPC is used, in which a tube around a nominal initial condition is found for which the robustness measure is guaranteed. Similarly, the control barrier functions in \cite{Lindemann2019} provide a forward invariant set around the initial condition.

In this work, we utilize genetic programming (GP) \cite{Koza1992} and reachability analysis \cite{Althoff2010} to synthesize controllers.
The benefit of genetic programming is that it is able to automatically find a structure for the controller, as the right structure is typically unknown beforehand \cite{Belta2019}. \rev{Moreover, the resulting controllers can be verified using off-the-shelf verification methods and are generally easier to interpret than e.g. neural network controllers or look-up tables in the form of binary decision diagrams (BDDs).} Genetic programming has been used for formal synthesis for reach-avoid problems in \cite{Verdier2018,Verdier2020}, in which controllers and Lyapunov-like functions are automatically synthesized for nonlinear and hybrid systems. Also, reachability analysis has been used in formal controller synthesis for reach-avoid problems, e.g.,  in \cite{Schurmann2018}, MPC is combined with reachability analysis, whereas in \cite{Schuermann2017c,Schurman2020,Ding2011} synthesizes a sequence of optimal control inputs \cite{Ding2011} or linear controllers \cite{Schuermann2017c,Schurman2020,schurmann2017optimal} for a sequence of time intervals.

Regardless, to the best of our knowledge, there are no closed-form controller synthesis methods which guarantee general STL specifications for a set of initial conditions. The goal of this work is to synthesize correct-by-construction closed-form controllers for nonlinear continuous-time systems subject to bounded disturbances for \rev{finite-time} STL specifications. Moreover, we consider a sampled-data implementation of the controller, i.e., the controller output is only updated periodically and is held constant between sampling times. To this end, we propose a framework based on counterexample-guided inductive synthesis (CEGIS) (see e.g. \cite{Solar2006,Ravanbakhsh2015,Abate2017,Raman2015}), combining model checking for STL \cite{Roehm2016}, the recent development of counterexample generation using reachability analysis \cite{Kochdumper2019b}, and genetic programming (GP) \cite{Koza1992}. This CEGIS approach combines a learning step with a formal verification step, in this case GP and reachability analysis, respectively. Within this framework, violations obtained during verification are used to improve the learning process, until a controller which formally satisfies the desired specification is found, \rev{or a user-defined maximum of iterations is met}. \rev{The synthesis of a closed-form sampled-data controller for general STL specifications is NP-complete. Unsurprisingly, the proposed method is not a complete method, i.e. existence of the solution does not guarantee that a solution will be returned in a finite number of iterations. Moreover, as the method relies on simulations, reachability analysis, and SMT solvers, the (offline) computational complexity of the proposed method is significant. However, the high computation time is offline and the method results in an interpretable closed-form sampled-data controller, both enabling digital implementation, that has a small online computational cost and small memory footprint.}

The main contributions of this work are twofold: first of all, we propose a CEGIS framework combining genetic programming with reachability analysis for the synthesis of closed-form sampled-data controllers for STL specifications. To enable reasoning over reachable sets as opposed to singular trajectories, \cite{Roehm2016} introduced reachset temporal logic (RTL) and proposed a sound transformation from STL to RTL. Our second contribution is the definition of quantitative semantics for RTL, and proving that the quantitative semantics is sound and complete. Similar to the quantitative semantics of STL, these quantitative semantics provide a measure of how robustly a formula is satisfied.

\section{Preliminaries}
The set of real positive numbers is denoted by $\R_{\geq 0}$. The power set of a set $S$ is denoted by $2^S$. Finally, an $n$-dimensional zero vector is denoted by $\mathbf{0}_n$. 

\subsection{Signal temporal logic}

We consider specifications expressed in signal temporal logic (STL) \cite{maler2004}, using the following grammar:
\begin{equation}
\varphi := \true \mid   h(s) \geq 0 \mid \neg \varphi \mid \varphi_1 \wedge \varphi_2 \mid \varphi_1 \mathcal{U}_{[a,b]} \varphi_2,
\label{eq:stlgrammar}
\end{equation}
where $\varphi$, $\varphi_1$, $\varphi_2$ are STL formula, and $h(s) \geq 0$ is a predicate over a signal $s: \R_{\geq 0} \rightarrow \R^n$ and a function $h: \R^n \rightarrow \R$. The Boolean operators $\neg$ and $\wedge$ denote negation and conjunction, respectively, and $\mathcal{U}_{[a,b]}$ denotes the bounded \textit{until} operator, i.e. until between $a$ and $b$, \rev{where $a<b$ and $a,b \in \mathbb{Q}_{\geq 0}$. Note that since $a,b \in \mathbb{Q}_{\geq 0}$, the STL formula inherently reasons over finite-time signals.} We can also define other standard (temporal) operators from \eqref{eq:stlgrammar}, such as disjunction $\varphi_1 \vee \varphi_2 := \neg ( \neg \varphi_1 \wedge \neg \varphi_2)$, next $\next_a \varphi:= \true ~ \mathcal{U}_{[a,a]} \varphi$, eventually $\event_{[a,b]} \varphi := \true ~\mathcal{U}_{[a,b]} \varphi$, and always 
$\always_{[a,b]} \varphi := \neg \event_{[a,b]} \neg \varphi$. The satisfaction relation $(s,t) \models \varphi$ indicates that the signal $s$ starting at $t$ satisfies $\varphi$. We consider the same definition of the semantics as in \cite{Roehm2016}, which slightly deviates from e.g. \cite{maler2004} w.r.t. the \textit{until} operator\footnote{In contrast to \cite{maler2004}, in our definition of the \textit{until} operator, $\varphi_1$ and $\varphi_2$ do not have to hold simultaneously}. Since we build upon the results in \cite{Roehm2016}, we have adopted the corresponding definition.
STL is equipped with \textit{quantitative semantics} $\rho(s,\varphi,t)$ that provides a robustness measure of how well a signal $s$ starting at time $t$ satisfies or violates the STL specification \cite{Donze2010,Fainekos2009}. 
If $\rho(s,\varphi,t)$ is negative, lower values imply that $\varphi$ is more strongly violated. Conversely, if $\rho(s,\varphi,t)$ is positive, higher values  imply that $\varphi$ is satisfied more robustly. 

\subsection{Reachset temporal logic}

\label{sec:RTLintro}
Consider a closed-loop system described by:
\rev{
\begin{equation}
\Sigma = \left\{ \begin{array}{rl}
\dot{\xi}(t) = & f_\mathrm{cl}(t,\xi(t),\omega(t)), \\
\xi(0) \in &I, ~\omega(t) \in \Omega,
\end{array} \right.
\label{eq:syscl}
\end{equation}
where $\xi(t) \in \R^n$ denotes the state, $\omega(t) \in \Omega \subset \R^l$ an external disturbance, $I \subset \R^n$ is the set of initial conditions, $I$ and $\Omega$ are compact, and $f:\R_{\geq 0} \times \R^n \times \R^l \rightarrow \R^n$ and $\omega:\R_{\geq 0} \rightarrow \R^l$ are assumed to be Lipschitz continuous.} 
In this work, we are not only interested in the STL performance of a singular trajectory, but rather of the set of all trajectories satisfying system $\Sigma$, defined by
\begin{equation}
\mathcal{S}(\Sigma) := \{\xi: \R_{\geq 0} \rightarrow \R^n \mid \forall t \geq 0: \xi(t) \text{ satisfies } \Sigma \}. 
\end{equation}
Now let us define the reachable set:
\begin{defn}[Reachable set]
Given a system $\Sigma$,
a mapping $R_e: \R_{\geq 0} \rightarrow 2^{\R^n}$ is an \textit{exact }reachable set if and only if: 
\begin{equation}
\forall t \in \R_{\geq 0} :  \{\xi(t) \mid \xi \in \mathcal{S}(\Sigma) \} = R_e(t).
\label{eq:reachsetcon}%
\end{equation}%
A mapping $R: \R_{\geq 0} \rightarrow 2^{\R^n}$ is a reachable set if and only if $\forall t \in \R_{\geq 0}:  R_e(t) \subseteq R(t)$.
\end{defn}%
That is, a reachable set satisfies that $\forall t \in \R_{\geq 0}, \forall \xi \in \mathcal{S}(\Sigma): \xi(t) \in R(t)$.
Reachability analysis tools such as CORA \cite{Cora} can return a sequence of sets
$
\Rs =\Rs_{\{t_0\}}  \Rs_{(t_0,t_1)} \Rs_{\{t_1\}} \Rs_{(t_1,t_2)} \dots  \Rs_{\{t_\mathrm{f}\}},
$
forming a reachable set given by
\begin{equation}
R(t) = \left\{ \begin{array}{rl} 
 \Rs_{\{t_i\}} &\text{if } t= t_i, \\
  \Rs_{(t_i,t_{i+1})} & \text{if } t\in (t_i,t_{i+1}) .
\end{array}
\right.
\label{eq:reachableset}
\end{equation}
The STL semantics over singular trajectories does not directly translate to the evaluation over reachable sets. To be able to reason directly over a reachable set, \cite{Roehm2016} introduced reachset temporal logic (RTL). The RTL fragment relevant for this work is given by:
\begin{align*}
\psi &:= \true \mid h(x) \geq 0 \mid \neg \psi \mid \psi_1 \wedge\psi_2, \\
\phi &:=  \mathcal{A} \psi \mid  \phi_1 \vee \phi_2 \mid \phi_1 \wedge \phi_2 \mid \next_a \phi.
\end{align*}
Here, $\psi$, $\psi_1$, $\psi_2$ are propositional formulae over states $x$ and $\phi$, $\phi_1$, $\phi_2$ formulae over a reachable set $R:\R_{\geq 0} \rightarrow 2^{\R^n}$. Additionally, $\mathcal{A}$ denotes the \textit{all} operator. The semantics is defined as follows:
\begin{align*}
x \models & h(x) \geq 0  &\iff & h(x) \geq 0,\\
x \models & \neg \psi & \iff & x \not \models \psi, \\
x \models & \psi_1 \wedge \psi_2 & \iff & x \models \psi_1 \text{ and } x \models \psi_2, \\
(R,t) \models & \mathcal{A} \psi  &\iff & \forall x \in R(t): x \models \psi, \\
(R,t) \models & \phi_1 \vee \phi_2 & \iff & (R,t) \models \phi_1 \text{ or } (R,t) \models \phi_2, \\
(R,t) \models & \phi_1 \wedge \phi_2 & \iff & (R,t) \models \phi_1 \text{ and } (R,t) \models \phi_2, \\
(R,t) \models & \next_a \phi & \iff & (R,t+a) \models \phi.
\end{align*}
\rev{Consider the following notion:
\begin{defn}[$c$-divisible]
\label{ass:cdiv}
An STL formula $\varphi$ is said to be $c$-divisible, if all interval bounds of the temporal operators of $\varphi$ are divisible by $c$.
\end{defn}
Note that since $a,b \in \mathbb{Q}_{\geq 0}$, there always exists a $c$ such that an STL formula is $c$-divisible.}
Given a $c$-divisible STL formula $\varphi$, the results in \cite[Lemma 2 \& Lemma 4]{Roehm2016} provide a sound transformation $\Upsilon$ to transform STL to RTL:
\begin{theorem}[Sound transformation {\cite[Theorem 1]{Roehm2016}}]
\label{thm:stl2rtl}
Given the system $\Sigma$ in \eqref{eq:syscl}, let $\varphi$ be a $c$-divisible STL formula, and $R(t)$ be the reachable set of $\Sigma$ in the form of \eqref{eq:reachableset} with $t_{i+1}-t_i = c$. The transformation $\Upsilon$ from \cite{Roehm2016}, bringing the STL formula $\varphi$ into an RTL formula $\phi = \Upsilon(\varphi)$, is sound, i.e.:
\begin{equation}
\forall \xi \in \mathcal{S}(\Sigma): (\xi,t) \models \varphi  \impliedby (R,t) \models \phi.
\label{eq:thm1implication}
\end{equation}
\end{theorem}
\rev{
The reachable set in \eqref{eq:reachableset} is formed by the reachable sequence $\mathcal{R}$, which partitions time into an alternating sequence of points and open intervals, whereas STL reasons over an infinite (but bounded) set of time instances. The transformation $\Upsilon$ first transforms the STL formula (rewritten in negation normal form) into sampled-time STL \cite{Roehm2016}: a subclass of STL that restricts to formulas with operators only reasoning over intervals $(0,c)$ and time shifts of fixed length $c$, such that the STL formula only reasons over an alternating sequence of points and open intervals. Here, the value $c/2$ can be seen as the time step between the points and a time interval. This transformation is sound, but not complete, i.e. for an STL formula $\varphi$ and transformed sampled-time STL formula $\varphi'$, we have $(\xi,t) \models \varphi) \impliedby (\xi,t) \models \varphi'$, but the converse is not necessarily true. Subsequently, sampled-time STL (in conjunctive normal form) is transformed into RTL. In this transformation, the reasoning over trajectories is replaced with reasoning over a reachable set. The transformation between sampled-time STL to RTL is sound and complete. The transformation from STL to sampled-time STL results in general in an over-approximation, which can be reduced by taking smaller values of $c$. Since the full definition of the transformation $\Upsilon$ is quite involved, we refer the interested reader to \cite{Roehm2016}.
}

The transformation $\Upsilon$ yields RTL formulae of the form
\begin{equation}
\phi = \bigwedge_{i\in \itset{I}} \bigvee_{j \in \itset{J}_i} \next_{j \frac{c}{2}} \bigvee_{k \in \itset{K}_{ij}} \mathcal{A} \psi_{ijk},
\label{eq:RTLstandard}
\end{equation}
where $\itset{I}, \itset{J}_{i}, \itset{K}_{ij}$ are finite index sets and $\psi_{ijk}$ are non-temporal subformulae. As can be seen, $j$ relates to a time step $c/2$, whereas $i$ and $k$ relate to the number of conjunctions and disjunctions. As example, the transformation of $\varphi = \varphi_x \mathcal{U}_{[c,2c]} \varphi_v$ with  $\varphi_x = x \geq 0$, $\varphi_v = v \geq 0$ is given by
\begin{align*}
\psi =& \next_0 \mathcal{A} \psi_x
    \wedge \next_{\frac{c}{2}} \mathcal{A} \psi_x 
     \\
    & \wedge \next_{c}\mathcal{A}( \psi_x \vee \psi_v)
    \wedge 
    (\next_{c} \mathcal{A} \psi_v \vee \next_{\frac{3c}{2}} \mathcal{A} \psi_x) \\
    &\wedge (\next_{c} \mathcal{A} \psi_v \vee \next_{\frac{3c}{2}} \mathcal{A} \psi_v \vee \next_{2c} \mathcal{A} (\psi_x \vee \psi_v)), 
\end{align*}%
with $\psi_x = x \geq 0$, $\psi_v = v \geq 0$. 
\subsection{Genetic programming}
 The controllers in this work are synthesized using genetic programming (GP) \cite{Koza1992}, a variant of genetic algorithms (GA) \cite{Holland1975}, which evolves entire programs rather than optimizing parameters. In our case, the evolved program is a controller based on elementary building blocks consisting of state variables and basic functions such as addition and multiplication. Within genetic programming, a candidate solution, called an \textit{individual}, is represented by a data structure enabling easy manipulation, such as an expression tree. This data structure is called the \textit{genotype}, whereas the individual itself, e.g., an analytic function, is referred to as the \textit{phenotype}. A pool of individuals, called the \textit{population}, is evolved based on a cost function, called the \textit{fitness} function, which assigns a fitness score to all individuals. Depending on the fitness score, individuals can be selected to be recombined or modified using \textit{genetic operators}, such as crossover and mutation. In the former, two subtrees of individuals are interchanged, whereas in the latter, a random subtree is replaced by a new random subtree. Each genetic operator has a user-defined rate, which determines the probability of the operator being applied to the selected individuals. A number of individuals are selected based on tournament selection: a fixed number of individuals are randomly selected from the population, and the individual with the highest fitness is returned. The process of selection and modification through genetic operators is repeated until a new population is created. The underlying hypothesis is that the average fitness of the population increases over many of these cycles, which are referred to as \textit{generations}. The algorithm is terminated after a satisfying solution is found or a maximum number of generations is met. 

\begin{figure}[t]%
\centering
\subfloat[Grammar]{\includegraphics[width = 0.4 \textwidth]{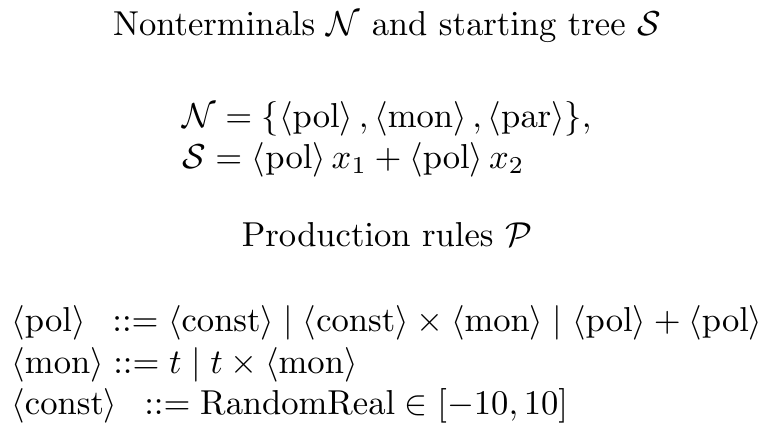} \label{fig:gram}} \hspace{1cm}
\subfloat[Genotype]{\includegraphics[width = 0.3\textwidth ]{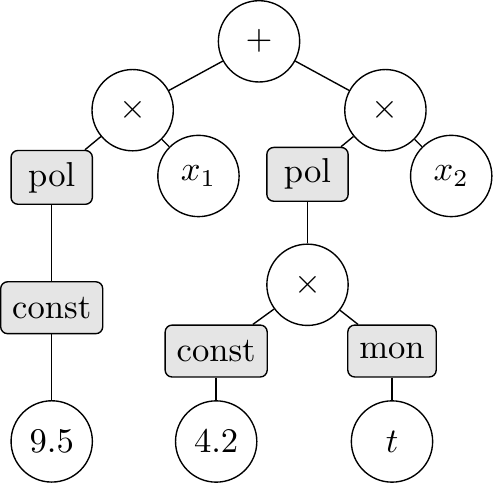} \label{fig:geno}} 
\caption{Example of a grammar and a genotype adhering to it. The corresponding phenotype is given by $9.5 x_1 +4.2t x_2$.}
\label{fig:grammar_examp}
\end{figure}

We use the variant grammar-guided genetic programming (GGGP) \cite{Whigham1995,Verdier2018}, which utilizes a grammar to which all individuals adhere: the population is initialized by creating random individuals adhering to the grammar and the used genetic operators are defined such that the resulting individuals also adhere to the grammar. The grammar is defined by the tuple $(\mathcal{N},\mathcal{S}, \mathcal{P})$, where $\mathcal{N}$ is a set of nonterminals, $\mathcal{S}$ a starting tree, and $\mathcal{P}$ a set of production rules, which relate nonterminals to possible expressions. An example of a grammar is shown in Figure \ref{fig:gram}. In this grammar, the nonterminals correspond to polynomials $\rmm{pol}$, monomials $\rmm{mon}$ over time $t$, and constants $\rmm{const}$. The starting tree $\mathcal{S}$ restricts the class of controllers to time-varying state feedback laws, linear in the state $x \in \R^2$.
Given the grammar in Figure \ref{fig:gram}, an example of a genotype is shown in Figure \ref{fig:geno}, which has the corresponding phenotype of $9.5 x_1 + 4.2 t x_2$. 

\section{Problem definition and solution approach}
\label{sec:PdSa}
We consider nonlinear systems subject to disturbances of the form:
\rev{
\begin{equation}
\Sigma_\mathrm{ol} = \left\{ \begin{array}{rl}
\dot{\xi}(t) = & f(t,\xi(t),u(t),\omega(t)), \\
\xi(0) \in &I, ~\omega(t) \in \Omega,
\end{array} \right.
\label{eq:syst}
\end{equation}
with states $\xi(t) \in \R^n$, inputs $u(t)\in \R^m$, bounded disturbances $\omega(t) \in \Omega \subset \R^l$, $I \subset \R^n$ is the set of initial conditions, $I$ and $\Omega$ are compact, and $f:\R_{\geq 0} \times \R^n \times \R^m \times \R^l \rightarrow \R^n$, $u:\R \rightarrow \R^m$, and $\omega:\R \rightarrow \R^l$ are assumed to be Lipschitz continuous.} 
We consider sampled-data time-varying state-feedback controllers $\kappa: \R_{\geq 0} \times \R^n \rightarrow \R^m$ such that 
\begin{equation}
    u(t) = \kappa(t_k,\xi(t_k)) \text{ for all } t \in [t_k, t_k+ \eta),
    \label{eq:SDcontroller}
\end{equation}
where $t_k$ denotes the $k$-th sampling instant, $
t_0 = 0$, and $\eta$ is the sampling time. This results in a closed-loop system of the form \eqref{eq:syscl} with, $\forall t\in [t_k, t_k+\eta)$,\rev{
\begin{equation}
f_\mathrm{cl}(t,\xi(t),\omega(t)) =f(t,\xi(t),\kappa(t_k,\xi(t_k)),\omega(t)).
\end{equation}}
The goal of this paper is formalized in the following:
\begin{prob}
\label{prob:1}
\rev{Given a $c$-divisible STL formula $\varphi$, the open-loop system \eqref{eq:syst}, and a sampling time $\eta$}, synthesize a closed-form sampled-data time-varying controller $\kappa: \R_{\geq 0} \times \R^n \rightarrow \R^m$ such that for all initial conditions and disturbances the resulting trajectories $\xi$ of the closed-loop system satisfy $\varphi$, i.e.:
\begin{equation}
\forall \xi \in \mathcal{S}(\Sigma): (\xi,0) \models \varphi
\label{eq:spec}
\end{equation}
\end{prob}

\begin{figure}[t]
\centering
\includegraphics[width = 0.7\linewidth]{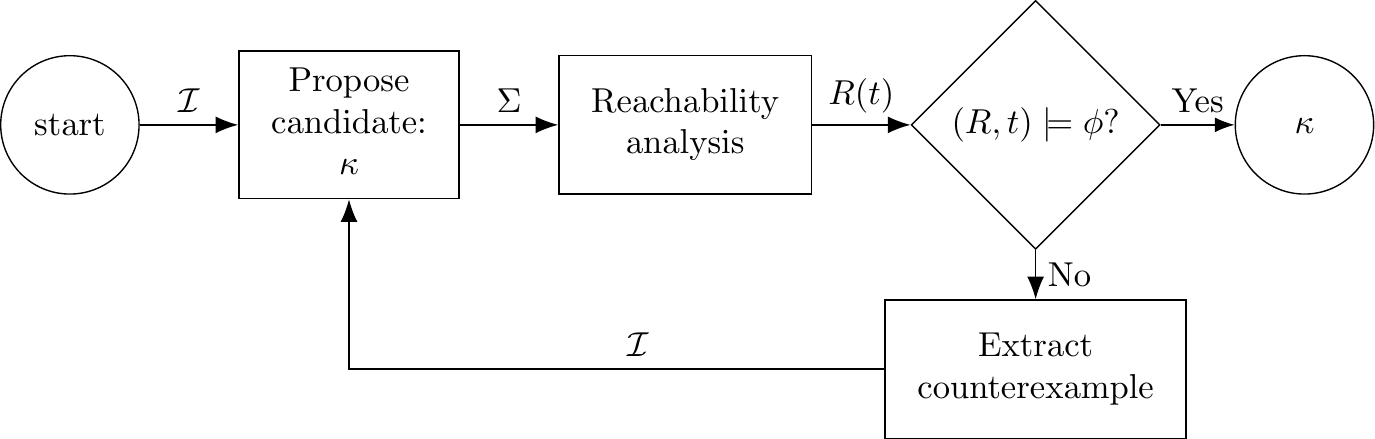}
\caption{Schematic overview of the algorithm. 
}
\label{fig:diagram}
\end{figure}
In this work, we propose a counterexample-guided inductive synthesis (CEGIS) framework to synthesize a controller such that $(R,0) \models \Upsilon(\varphi)$, thereby solving solving Problem \ref{prob:1} as follows from Theorem \ref{thm:stl2rtl}. The framework consists of iteratively proposing a controller obtained through GGGP\footnote{While GGGP evolves a population of controllers, only the controller with the highest fitness is returned.} and then formally verifying the RTL formula $\Upsilon(\varphi)$ using reachability analysis. The proposed controller by GGGP is optimized w.r.t. a set of simulated trajectories \rev{(obtained through numerical integration)}, with the underlying idea that these are relatively fast to compute and provide a sensible search direction for the synthesis. The \rev{computationally more intensive} reachability analysis verifies the resulting controller.

For a given open-loop system $\Sigma_\mathrm{ol}$, STL formula $\varphi$, and grammar $(\mathcal{N},\mathcal{S},\mathcal{P})$, the algorithm is initialized as follows:
\begin{enumerate}
\item[I1)] The RTL formula $\phi$ is computed using $\phi = \Upsilon(\varphi)$ (see Theorem \ref{thm:stl2rtl}).
\item[I2)] The set $\mathcal{I}$, consisting of pairs of initial conditions and disturbance realizations is initialized by randomly choosing $n_\mathrm{s}$ initial conditions $\{x^1, \dots, x^{n_\mathrm{s}}\} \subset I$, with random disturbance realizations $\omega^i :\R_{\geq 0} \rightarrow \Omega$, such that $\mathcal{I} = \left\{\left(x^1, \omega^1 \right), \dots, \left(x^{n_\mathrm{s}},\omega^{n_\mathrm{s}} \right)   \right\}$.
\end{enumerate}
Given the initialized data, the algorithm goes through the following cycle, illustrated in Figure \ref{fig:diagram}, where each cycle is referred to as a \textit{refinement}:
\begin{enumerate}
\item[A1)] A candidate solution is proposed using GGGP, based on simulation trajectories corresponding to the set $\mathcal{I}$. 
\item[A2)]  For the given candidate controller, the reachable set is computed.
\item[A3)]  Based on the reachable set, either:
\begin{enumerate}
\item $(R,t) \models \phi$, \rev{which is formally verified through SMT solvers}, thus a controller solving Problem \ref{prob:1} is found.
\item $(R,t) \not\models \phi$, and a counterexample is extracted in the form of an initial condition $x$ and a corresponding disturbance realization $w$. This pair $(x,w)$ is added to $\mathcal{I}$ and the algorithm returns to step A1).
\item $(R,t) \not\models \phi$ and a maximum of refinements is reached, therefore the algorithm is terminated.
\end{enumerate}
\end{enumerate}
To quantify the violation or satisfaction of an RTL formula, we introduce quantitative semantics for RTL in the next section. The proposal of a candidate controller in step A1) is discussed in Section \ref{sec:proposal}. The verification and counterexample generation in step A3) is discussed in Section \ref{sec:verification}.

\section{Quantitative semantics}
\label{sec:RTLrobust}
Inspired by the quantitative semantics of STL \cite{Donze2010,Fainekos2009}, we define quantitative semantics for RTL in this section. These quantitative semantics provide a \textit{robustness measure} on how well the formula is satisfied. For an RTL formula $\phi$ with propositional subformulae $\psi$, the quantitative semantics is given by functions $P(R,\phi,t)$ and $\varrho(x,\psi)$, respectively, recursively defined as:
\begin{align*}
\varrho(x, \true) = &+\infty,\\
\varrho(x, h(x) \geq 0) = & h(x), \\
\varrho(x, \neg \psi ) = & -\varrho(x,\psi), \\
\varrho(x,  \psi_1 \wedge \psi_2 ) =& \min(\varrho(s, \psi_1), \varrho(s, \psi_2)),\\
P(R, \mathcal{A} \psi ,t) = & \inf_{x \in R(t)} \varrho(x,\psi), \\
P(R,  \phi_1 \vee \phi_2 ,t) = &\max(P(R, \phi_1 ,t), P(R, \phi_2,t)),\\
P(R,  \phi_1 \wedge \phi_2 ,t) =& \min(P(R, \phi_1 ,t), P(R, \phi_2,t)),\\
P(R, \next_{a} \phi ,t) = & P(R, \phi ,t+a).
\end{align*}
The quantitative semantics of STL are sound and complete \cite{Fainekos2009,Donze2013}. The quantitative semantics of RTL also have these properties:
\begin{theorem}[Soundness and completeness]
\label{thm:soundness}
Let $\phi$ be an RTL formula, $R$ a reachable set, and $t$ a time instance, then:
\begin{enumerate}
\item[1)] $P(R, \phi,t) > 0 \Rightarrow (R,t) \models \phi$ and $ (R,t) \models  \phi \Rightarrow  P(R, \phi,t)  \geq  0,$
\item[2)] $P(R, \phi,t) < 0 \Rightarrow (R,t) \not\models \phi$ and $(R,t) \not\models \phi  \Rightarrow  P(R, \phi,t)  \leq 0.$
\end{enumerate}
\end{theorem}
\begin{rem}
Note that $P(R,\phi,t) = 0$ does not imply $(R,t) \models \phi$ nor $(R,t) \not\models \phi$. This is because on the boundary of an inequality, the distinction between inclusion or exclusion is lost within the quantitative semantics. That is, if  $\varrho(x,\psi) = 0$, we also have $\varrho(x,\neg \psi) =0$, hence the quantitative semantics of two mutually exclusive logic formulae evaluate to the same value.
\end{rem}

The proof of Theorem \ref{thm:soundness} can be found in Appendix \ref{app:proofs}. Consider a $c$-divisible STL formula $\varphi$ and the corresponding RTL formula $\phi = \Upsilon(\varphi)$ in the form of \eqref{eq:RTLstandard}. Using the equivalences $\next_a (\phi_1 \wedge \phi_2) = \next_a \phi_1 \wedge \next_a \phi_2$ and rewriting $\psi_{ijk}$ in disjunctive normal form, we can express the RTL formula \eqref{eq:RTLstandard} as:
\begin{subequations}
\begin{align}
\phi &= \bigwedge_{i \in \itset{I}} \bigvee_{j \in \itset{J}_i,k \in \itset{K}_{ij}} \phi' _{ijk}, \label{eq:maincheck} \\
\phi'_{ijk} &=  \next_{j \frac{c}{2}} \mathcal{A}   \bigvee_{a \in \itset{A}^{ijk}} \bigwedge_{b \in \itset{B}^{ijk}_a} h^{ijk}_{ab}(x) \sim 0, \label{eq:subformcheck}
\end{align}%
\label{eq:RTLcheck}%
\end{subequations}%
where $\itset{A}^{ijk}$ and $\itset{B}^{ijk}_a$ denote finite index sets, $\sim \in \{ \geq, >\}$, and $h^{ijk}_{ab}(x) \sim 0$ is a predicate over $x$. Using the quantitative semantics defined in Section \ref{sec:RTLrobust}, the robustness measure of this RTL formula is given by
\begin{subequations}
\begin{align}
P(R,\phi,0) &= \min_{i\in \itset{I}} \left( \max_{j \in \itset{J}_i,k \in \itset{K}_{ij}} P(R,\phi'_{ijk},0) \right), \label{eq:robustfull}\\
P(R,\phi'_{ijk},0)& = \! \! \!\! \inf_{x \in R  \left(j\frac{c}{2}\right)}\!\! \left( \max_{a\in \itset{A}^{ijk}} \left( \min_{b \in \itset{B}^{ijk}_a}  h^{ijk}_{ab}(x) \right) \! \!  \right) \! .\label{eq:subformRD}
\end{align}%
\label{eq:RTLrobust}%
\end{subequations}%

\section{Candidate controller synthesis}
\label{sec:proposal}
In this section, we detail step A1) of the proposed algorithm in Section \ref{sec:PdSa}, i.e., 
the proposal of a candidate controller. The candidate controller is synthesized using GGGP, by maximizing an approximation of the robustness measure, which is based on a finite number of simulated trajectories, \rev{obtained through numerical integration}.
The sampling time is equal to $c/2$ to coincide with the time instances at which the robustness measure $P(R,\phi,0)$ is evaluated. 
For an RTL formula of the form \eqref{eq:RTLstandard}, the first and the final time instances of relevance $\tau_0$ and $\tau_\mathrm{f}$, are given by $\tau_0 = 0$ and $\tau_\mathrm{f} = \frac{c}{2} \max_{i \in \itset{I}} | \itset{J}_i|$, respectively. \rev{Let us denote the finite set of sampled-time instances $\hat{T} = \{\tau_0, \dots, \tau_\mathrm{f}\}$.}
Given a candidate controller $\kappa: \R_{\geq 0} \times \R^n \rightarrow \R^m$, a set of pairs of initial conditions and disturbance realizations $\mathcal{I}$, we consider an approximated reachable set \rev{$\hat{R}_\mathcal{I}^\kappa:\hat{T} \rightarrow 2^{\R^n}$} formed by all corresponding simulated trajectories \rev{$x: \hat{T} \rightarrow \R^n$}, \rev{such that for a given time instance $\tau_q \in \hat{T}$:}
\begin{equation*}
\begin{array}{r}
\hat{R}_\mathcal{I}^\kappa(\tau_q) \! = \!
\{ x(\tau_q) \mid (x(\tau_0), \omega) \in \mathcal{I}\}.
\end{array}
\end{equation*}
Provided this set $\hat{R}_\mathcal{I}^\kappa(\tau_q)$, we approximate the robustness measure by $P(\hat{R}_\mathcal{I}^\kappa, \phi, 0)$. 

\subsection{Outline of the candidate controller synthesis}
\label{sec:proposaloutline}
\rev{
The proposal of a candidate controller in step A1) is based on approximating an optimal controller that solves:
\begin{equation}
    \sup\limits_{\kappa} \inf\limits_\mathcal{I}~ P(\hat{R}_\mathcal{I}^\kappa,\phi,0).
\end{equation}
If the optimum is positive, it follows from Theorem \ref{thm:soundness} and Theorem \ref{thm:stl2rtl} that the corresponding optimal controller $\kappa^*$ solves Problem \ref{prob:1}. To (approximately) solve this optimization problem, the algorithm alternatively updates the controller and the disturbances within $\mathcal{I}$, as described in the following steps, which are also illustrated in Figure \ref{fig:diagram2}:
}
\begin{Description}[margin =A1.a)]
\item[A1.a)] \rev{Given the set $\mathcal{I}$}, We synthesize an analytic expression $\kappa:\R\times \R^n \rightarrow \R^m$ by using GGGP to solve:
\begin{equation}%
 \argsup\limits_{\kappa}~ P(\hat{R}_\mathcal{I}^\kappa,\phi,0).
 \label{eq:GGGPcontroller}%
\end{equation}%
If for the resulting controller $\kappa$ the robustness measure approximation $P(\hat{R}_\mathcal{I}^\kappa, \phi,0)$ is negative, this optimization step in \eqref{eq:GGGPcontroller} is repeated. Otherwise, the algorithm continues to the next step.
\item[A1.b)] \rev{Given the controller $\kappa$}, for each initial condition $x^{i}$ in $\mathcal{I}$, an analytic expression for a disturbance realization $\omega^i:\R\rightarrow \Omega$ is synthesized using GGGP, in which the robustness measure approximation is minimized, i.e.:
\rev{%
\begin{align}
\begin{array}{ll}
\omega^* = \argsup\limits_{\omega^i}& -P\left(\hat{R}_{\{(x^{i},\omega^i)\}}^\kappa,\phi,0\right).
\end{array}
\label{eq:distopt}
\end{align}}%
\rev{
The set $\mathcal{I}$ is then updated by replacing $(x^i,\omega^i)$ with $(x^i,\omega^*)$, i.e.
\begin{equation*}
    \mathcal{I} \leftarrow (\mathcal{I}\backslash \{(x^i,\omega^i)\}) \cup \{(x^i, \omega^*)\}.
\end{equation*}}%
If the corresponding robustness degree approximation $P(\hat{R}_{\{(x^{i},\omega^i)\}}^\kappa, \phi,0)$ is negative, the algorithm returns to step A1.a). Otherwise, if for all updated disturbance realizations the robustness measure approximation is positive, i.e., $\forall i,~ P(\hat{R}_{\{(x^i, \omega^i)\}}^\kappa, \phi,0) > 0$, the algorithm returns a candidate controller.
\end{Description}

\begin{figure}[t]
\centering
\includegraphics[width = 0.7\linewidth]{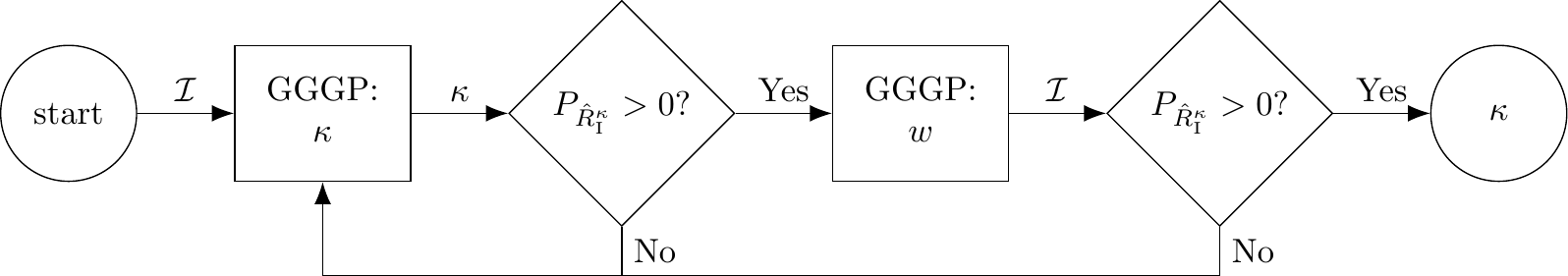}
\caption{Schematic overview of the synthesis of candidate controller. 
}
\label{fig:diagram2}
\end{figure}

\subsection{Reference-tracking controllers}
\label{sec:reftrack}
To speed up the synthesis, \rev{it is possible to impose a structure to the solution. In this section we discuss the design of reference-tracking controllers,} based on a nominal reference trajectory $x_\mathrm{ref}(t)$ and a corresponding feedforward input $u_\mathrm{ff}(t)$. That is, we consider a time-varying reference-tracking controller of the form:
\rev{
\begin{equation}
\kappa(t,x)  = u_\mathrm{ff}(t) + \kappa_\mathrm{fb}(t,x-x_\mathrm{ref}(t)),
\label{eq:con_struct}
\end{equation}}%
where $\kappa_\mathrm{fb}: \R_{\geq 0} \times \R^n \rightarrow \R^m$ is a time-varying feedback controller. \rev{This controller is then used in a sampled-data fashion as in \eqref{eq:SDcontroller}}. The feedforward input and reference trajectory can be computed beforehand as follows:
\begin{enumerate}
\item[R1)] Given a point $x_0 \in int(I)$, (e.g. the centroid of $I$ if $I$ is convex), an analytic expression for $u_\mathrm{ff}: \R \rightarrow \R^m$ is synthesized using GGGP, by maximizing the approximated robustness measure for a nominal trajectory starting at $x_0$, i.e. a trajectory with no disturbance:
\begin{align*}
\begin{array}{ll}
\argsup \limits_{u_\mathrm{ff}} & P(\hat{R}_{\{ (x_0, \textbf{0}_l) \}}^{u_\mathrm{ff}}, \phi, 0).
\end{array}
\end{align*}
  \item[R2)] Given the feedforward input $u_\mathrm{ff}$, an analytic expression for the corresponding \rev{nominal} reference trajectory $x_\mathrm{ref}: \R \rightarrow \R^n$ is synthesized using GGGP. \rev{Given the simulated solution $x(\tau_k)$ corresponding to $x(0) = x_0$, $\omega(t) = \textbf{0}_l$, and $u(t) = u_\mathrm{ff}(t)$, $x_\mathrm{ref}$ is obtained by fitting for each state dimension $i \in \{1, \dots, n\}$ an expression to $x_i(\tau_k)$}, based on the Euclidean norm of the error vector $e_i = [e_i(\tau_0), \dots, e_i(\tau_\mathrm{f})]$, with $e_i(\tau_k) = x_i(\tau_k) -x_{\mathrm{ref},i}(\tau_k)$, i.e., maximizing:
\begin{equation*}
 \argsup\limits_{x_{\mathrm{ref},i}} (1-\|e_i\|)^{-1}.
\end{equation*}%
\end{enumerate}%
Using the synthesized pair $(u_\mathrm{ff}(t), x_\mathrm{ref}(t))$, the user-defined grammar within GGGP can be used to enforce the structure of a time-varying reference controller in \eqref{eq:con_struct} within step A1), as demonstrated by the following brief example:
\begin{exmp}
Let us consider a one-dimensional system with dimensions $n = m =1$. The structure of \eqref{eq:con_struct}, where we further restrict $\kappa$ to be linear in state, can be enforced by taking the starting tree $\mathcal{S} = u_\mathrm{ff}(t) + \rmm{pol}(x-x_\mathrm{ref})$ and the production rules from Figure \ref{fig:gram}.
\end{exmp}

\section{Reachability analysis and verification}
\label{sec:verification}
In this section we detail step A3) of the algorithm. \rev{In this work, we consider polynomial zonotopes $\mathcal{PZ}$ as the set representation of the reachable set\footnote{In this definition, without loss of generality and for the ease of exposition, we only consider (dependent) generators $G$ and we omit independent generators; for the full definition we refer to \cite{Kochdumper2019a}.}:
\begin{defn}[Polynomial zonotope]
\label{def:PZ}
Given a generator matrix $G \in \R^{n \times h}$ and exponent matrix $E \in \mathbb{Z}_{\geq 0}^{p \times h}$, a polynomial zonotope $\mathcal{PZ}$ is defined as
\begin{equation*}
\begin{array}{r}
\mathcal{PZ}  := \left\{ \sum_{i=1}^h \left( \Pi_{k=1}^p \alpha_k^{E_{(k,i)}} \right) G_{(\cdot,i)}  ~\middle|~ \alpha_k \in [-1,1] \right\}.
\end{array}
\end{equation*}
The vector $\bm{\alpha} = [\alpha_1,\dots,\alpha_p]^T$ is referred to as the \textit{parameterization vector} of the polynomial zonotope.
\end{defn}
Consider a parameterization vector $\bm{\alpha}$ and a reachable set $R(t)$ expressed as a polynomial zonotope. The corresponding point in the reachable set $z(\bm{\alpha},R(t)) \in \R^n$ is given by:
\begin{equation}
z(\bm{\alpha},R(t)) =  \sum_{i=1}^h \left( \Pi_{k=1}^p \alpha_k^{E_{R(k,i)}} \right) G_{R(\cdot,i)},
\end{equation}%
where $E_R$ and $G_R$ denote the exponent matrix and generator matrix of $R(t)$, respectively.
The benefit of polynomial zonotopes as set representation is that dependencies between points in subsequent reachable sets are maintained under the reachability analysis operations \cite{Kochdumper2019b}. That is, for a reachable set $R: \R_{\geq 0} \rightarrow 2^{\R^n}$ and parameterization vector $\bm{\alpha}$, we have $\xi(t) = z(\bm{\alpha},R(t)) \implies \xi(0) = z(\bm{\alpha},R(0))$}
This enables the extraction of an initial condition corresponding to a point for which the specification is violated. Using this method, we construct a counterexample in the form of a pair of initial condition and disturbance realization $(x_0,\omega)$, such that the corresponding trajectory results in a violation of the RTL formula. After reachability analysis, the algorithm undergoes the following steps:
\begin{enumerate}
\item[B1)] For all subformulae $\phi'_{ijk}$ in \eqref{eq:subformcheck}, the corresponding robustness sub-score \eqref{eq:subformRD} is computed by solving the following nonlinear optimization problem\footnote{
To use gradient-based optimization, $\max$ and $\min$ can be approximated by
$M^\beta_{a \in \itset{A}} (x_a) = (\sum_{a\in \itset{A}} x_a e^{\beta x_a })/( \sum_{a\in \itset{A}} e^{\beta x_a }),$
where $\itset{A}$ denotes an iterator set and for $\beta \rightarrow \infty$, $M^\beta_{a \in \itset{A}} (x_a) \rightarrow \max_{a \in \itset{A}} x_a$ and $\beta \rightarrow -\infty$, $M^\beta_{a \in \itset{A}} (x_a) \rightarrow \min_{a \in \itset{A}} x_a$.}over the corresponding set $R(jc/2)$:
\end{enumerate}
\rev{%
\begin{equation}
\begin{array}{c}
p_{ijk}^* \! = \!  \inf \limits_{\bm{\alpha}_{ijk}}  \!\left(  \max \limits_{a\in \itset{A}^{ijk}}  \!\left(  \min \limits_{b \in \itset{B}_a^{ijk}}   h^{ijk}_{ab}(z(\bm{\alpha}_{ijk},R\left(\frac{jc}{2}\right)) \! \right)\!\!\! \right) \!.\\
\end{array}
\label{eq:optRD}
\end{equation}}
\begin{enumerate}
\item[B2)] Given the robustness sub-scores $p_{ijk}^*$, compute the full robustness measure \eqref{eq:robustfull}:
\begin{equation}
p^*=\min_{i\in \itset{I}} \max_{j \in \itset{J}_i, k\in \itset{K}_{ij}} p^*_{ijk}.
\label{eq:robustnessopti}
\end{equation}
\item[B3)] As we rely on nonlinear optimization, we cannot guarantee to find the global optimum $p^*$, but rather an upperbound $\hat{p}$, such that $P(R,\phi,0) = p^* \leq \hat{p}$. Given $\hat{p}$, either:
\begin{enumerate}
\item $\hat{p}<0$, hence the RTL specification is violated. In this case, given the argument $(ijk)^*$ solving \eqref{eq:robustnessopti}, we  extract an initial condition $x_0$ corresponding to \rev{$\bm{\alpha}_{(ijk)^*}$, i.e. $x_0 = z(\bm{\alpha}_{(ijk)^*},R(0))$}. For this initial condition $x_0$, a disturbance realization $\omega$ is synthesized similarly to step A1.b), i.e., GGGP is used to solve:
\begin{equation*}
\argsup\limits_{\omega}-P\left(\hat{R}_{\{(x_0,\omega)\}}^\kappa,\phi,0 \right).
\end{equation*}
The pair $(x_0,\omega)$ is subsequently added to $\mathcal{I}$. This new set $\mathcal{I}$ is then used to improve upon the synthesized controller in step A1).
\item $\hat{p} \geq 0$, hence the RTL specification is potentially satisfied. However, to guarantee this, we perform an additional verification step, based on Satisfiability Modulo Theories (SMT) solvers \cite{Barrett2009}, which are capable of verifying first-order logic formulae. The subformula \eqref{eq:subformcheck} holds if the following first-order logic formula holds:
\begin{equation}
\forall x \in R \left(j \frac{c}{2}\right): \bigvee_{a \in \itset{A}^{ijk}} \bigwedge_{b \in \itset{B}^{ijk}_a} h_{ab}^{ijk}(x) \sim 0,
\label{eq:SMTsub}
\end{equation}
where again $\sim \in \{\geq, >\}$. Suitable SMT solvers to verify \eqref{eq:SMTsub} include Z3 \cite{Z3} when $R(j c/2)$ and $h_{ab}^{ijk}$ are expressed as polynomials, and dReal \cite{Gao2013} when these are expressed as general nonlinear expressions\footnote{
dReal implements a $\delta$-complete decision procedure \cite{Gao2010}. If the reachable set is robust w.r.t. the RTL formula, this has no consequence for our proposed framework.}. Given the Boolean answers to the subformulae in \eqref{eq:subformcheck} for all $ijk$, it is trivial to compute the Boolean answer to \eqref{eq:maincheck}. 
\end{enumerate}

\end{enumerate}

\rev{
A synthesized controller, formally verified in step B3.b), solves Problem \ref{prob:1}, as formalized in the following theorem:
\begin{theorem}[Correct-by-design controller]
Given a $c$-divisible STL formula $\varphi$, an open-loop system \eqref{eq:syst}, and a sampling time $\eta$, if the algorithm in Section \ref{sec:PdSa} returns a controller before the maximum number of refinements, then the closed-loop system satisfies
\begin{equation}
    \forall \xi \in \mathcal{S}(\Sigma): (\xi,0) \models \varphi.
\end{equation}
\end{theorem}

\begin{proof}
If the algorithm terminates, the returned controller results in a reachable set $R$ such that $(R,0) \models \psi$, where $\phi = \Upsilon(\varphi)$. By Theorem \ref{thm:stl2rtl}, we have $\forall \xi \in \mathcal{S}(\Sigma): (\xi,0) \models \varphi$.
\end{proof}
}
\section{Dealing with conservatism}
\label{sec:conservatism}

\rev{Conservatism, in the reachability analysis and the transformation from STL to RTL, makes possible that $(R,0) \not \models \phi$, whereas $\forall \xi(0) \in I$, $(\xi,0) \models \varphi$, i.e., the desired STL specification holds for all initial conditions, whereas based on the reachability set, the RTL specification is not met. This conservatism can be reduced refining settings such as the time steps or Taylor order in the reachability tool (see \cite{Althoff2013}), or reducing the parameter $c$ to obtain less conservative RTL formulae $\phi$, at the cost of increased overall computational complexity. Similarly, truncation errors of the integration scheme, and conservatism within reachability analysis (introducing spurious trajectories) can lead to mismatches between $\hat{R}_\mathcal{I}^\kappa$ and $R(t)$. This mismatch can be bridged considering an optional error signal $\varepsilon$ added to the simulated trajectory $x(\tau_q)$, which is co-synthesized with the disturbance realizations, as will be shown in the case studies in the next section.}

Issues due to conservatism can also be dealt with within the synthesis of a candidate controller in step A1), e.g. the controllers within GGGP could be further optimized w.r.t. the robustness, such that the added robustness could potentially compensate for conservative reachability analysis. Controller complexity (measured as the the number of nonterminals) can also be used as secondary optimization criteria to facilitate less conservative reachability analysis. The resulting multi-objective optimization problem is solved using Pareto-optimality ranking~\cite{NSGA} that results in a rank that is used as the new fitness value. 

\rev{Finally, note that the optimization problems in steps A1.a), A1.b), R1), R2), B1), and B3.a) are non-convex and therefore finding a global optimum cannot be guaranteed. The optimization problems are used to propose candidate controllers or to provide counterexamples that constrain the solution space. As the goal is to find a qualitatively correct controller rather than an (quantitatively) optimal one, loss in optimality is of a lesser importance. Moreover, by the use of mutation within genetic programming, the algorithm is capable of exploring the search space, until a solution is found.}

\section{Case studies}
\label{sec:casestudies}

\begin{table}
\centering
\caption{General settings for each of the case studies. The  number of individuals, GGGP generations and CMA-ES generations are shown for each controller component and disturbance realizations.}
\smallskip
\label{tab:settings}
\scalebox{0.9}{  
\begin{tabular}{rccccccccccccc}
  \toprule
 \multirow{2}{*}{\textbf{System}}& \multirow{2}{*}{${n_\mathrm{s}}$} & \multicolumn{4}{c}{\textbf{Individuals}} & \multicolumn{4}{c}{\textbf{GGGP generations}} & \multicolumn{4}{c}{\textbf{CMA-ES generations}}  \Tstrut \\
 \cmidrule(lr){3-6}  \cmidrule(lr){7-10}  \cmidrule(lr){11-14}
 & & $u_\mathrm{ff}$ & $x_\mathrm{ref}$ & $\kappa$ & $\omega^i$&  $u_\mathrm{ff}$ & $x_\mathrm{ref}$ & $\kappa$ &  $\omega^i$ & $u_\mathrm{ff}$ & $x_\mathrm{ref}$ & $\kappa$ & $\omega^i$  \Bstrut
 \\ \midrule 
Car & 7 &  14 & 14 & 14 & 14 &  30 & 10 & 3 &3 &20 & 10 & 10 & 3\Tstrut \\
Path planning & 10 & 28 & 28 & 14 & 14 & 30 & 50 & 3 & 3 & 40 & 40 & 10 & 3 \\
Aircraft & 5 & 28& 42 & 14& 14 & 50 & 50 & 5 & 5& 40 & 60 & 10 &  3\\
Platoon & 10 & - & - & 14 & 14 & - & - & 3 & 3 & - & - & 10 & 3\\
Spacecraft & 7 & - & - & 14 & 14 & - & - & 5 & 5 & - & - & 10 & 3\Bstrut \\
 \bottomrule
  \end{tabular}%
  }
\end{table}

\begin{table}
\centering
\caption{Production rules $\mathcal{P}$.}
\label{tab:PRexamp}
\centering
\smallskip
\scalebox{0.9}{  
  \begin{tabular}{rl}
  \toprule
  $\boldsymbol{\mathcal{N}}$&\textbf{ Rules}  \Tstrut\Bstrut \\ \midrule
$\rmms{expr}{t} $ & $::= \rmms{pol}{t} ~|~ \rmms{pol}{t} \times \rmms{trig}{t} ~|~ \rmms{expr}{t} + \rmms{expr}{t}$ \Tstrut \\ 
$\rmms{trig}{t}$ & $::= \tanh(\rmms{pol}{t}) ~|~ \sin(\rmms{pol}{t}) ~|~ \cos(\rmms{pol}{t})$ \\
$\rmms{pol}{t} $ & $::= 0 ~|~ \rmm{const}  ~|~ \rmm{const}\times\rmms{mon}{t}~|~ \rmms{pol}{t} + \rmms{pol}{t} $\\
$\rmms{mon}{t}$ &$::= t ~|~ t \times \rmms{mon}{t} $ \\
$\rmms{pol}{x} $ & $::= \rmm{const}\times\rmms{mon}{x}~|~ \rmms{pol}{x} + \rmms{pol}{x} $\\
$\rmms{mon}{x}$ &$::= \rmm{var} ~|~ \rmm{var} \times \rmms{mon}{x} $ \\
$\rmm{var}$ &$::= x_1 ~|~ \dots ~|~ x_n $ \\
 $\rmm{const} $ &$::=$ Random Real $\in \left[-1,1\right]$ \Bstrut \\
 \bottomrule
  \end{tabular}%
  }
\end{table}

\rev{In this section we demonstrate the effectiveness of the proposed framework on benchmarks from competing synthesis methods, i.e. reachability-based \cite{Schuermann2017c} (car example), MPC-based \cite{Lindemann2017} (path planning), and abstraction-based \cite{Reissig2017} (airplane landing manoeuvre). Moreover, we consider the effect of input saturation, which is enforced through the STL specification. Additionally, we consider a platooning benchmark \cite{schurmann2017optimal} to investigate the scalability w.r.t. system dimension. While we use for the aforementioned benchmarks the reference-tracking controller structure discussed in Section \ref{sec:reftrack}, we demonstrate the ability to synthesize controllers from scratch on a simplified spacecraft \cite{Brockett1983} in Section \ref{sec:discovering}. 
}

The case studies are performed using an Intel Xeon CPU E5-1660 v3 3.00GHz using 14 parallel CPU cores. The GGGP algorithm is implemented in Mathematica 12 and the reachability is performed using CORA in MATLAB. \rev{Motivated by the non-convex and discontinuous nature of the optimization problems, we use population-based optimization methods, but any suitable optimization tool can be used instead. For the optimization problem in \eqref{eq:optRD}, we use particle swarm optimization of the global optimization toolbox in MATLAB. Within each generation of GGGP, parameters within an individual are optimized using Covariance Matrix Adaptation Evolution Strategy (CMA-ES) \cite{hansen2001}, based on the same fitness function as used for GGGP. More specifically, we use the variant sep-CMA-ES \cite{Ros2008}, due to its linear space and time complexity.} For the verification of \eqref{eq:SMTsub}, we use the SMT solver dReal with $\delta = 0.001$.

Across all benchmarks, the probability rate of the crossover and mutation operators being applied on a selected individual are 0.2 and 0.8, respectively. Benchmark-specific settings are shown in Table \ref{tab:settings}, which include the number of simulations ${n_\mathrm{s}}$, number of individuals, and the number of GGGP and CMA-ES generations. Note that the number of GGGP generations for $\kappa$ and $\omega^i$ is the number of generations per step A1.a) and A1.b), and not the total of GGGP generations per proposal of a controller in step A1), which depends on the number of times step A1.a) and A1.b) are repeated. For each case study, we use a grammar with nonterminals and production rules as shown in Table \ref{tab:PRexamp}. \rev{These nonterminals correspond to general time-dependent expressions $\rmms{expr}{t}$, time-dependent trigonometric functions $\rmms{trig}{t}$, time- and state-dependent polynomial expression $\rmms{pol}{t}$ and $\rmms{pol}{x}$, respectively, time- and state-dependent monomials $\rmms{mon}{t}$ and $\rmms{mon}{x}$, respectively, variables $\rmm{var}$, and constants $\rmm{const}$. The polynomials are restricted to polynomials over either time $t$ or states $x$, where the state-dependent polynomials are further restricted to not contain zero degree monomials.
The time-dependent expressions are formed by time-dependent polynomials, a product of these polynomials, and time-dependent trigonometric functions, and a sum of two expressions. The trigonometric functions are restricted to hyperbolic tangents, sines and cosines with time-dependent polynomial arguments. Note that per case study, different starting trees are used, such that potentially only a subset of the grammar is available. E.g., if the starting tree is $\rmms{pol}{t}$, candidate solutions are restricted to time-dependent polynomial solutions.}

We use Runge-Kutta as numerical integration scheme. To keep a constant number of initial conditions in $\mathcal{I}$, counterexamples are added using a first-in, first-out principle. To compensate for the gap between the simulation and the reachability analysis (as discussed in Section \ref{sec:conservatism}), we consider an added error signal bounded by the scaled vector field of the dynamics $f$, parameterized by
\begin{equation}
\varepsilon(t,x) = \delta \sigma(t) f(t,x(t),u(t),\omega(t)),
\label{eq:errorsignal}
\end{equation}
where $\delta$ is a constant and $\sigma: \R_{\geq 0} \rightarrow [-1,1]^{n \times n}$ a time-varying diagonal matrix which determines the sign and magnitude of the error signal. The constant $\delta$ is optimized after each reachability analysis such that the mismatch between the robustness measure and the approximated robustness measure is minimized, i.e.:
\begin{equation}
\arginf_\delta \left\|P(R,\phi,0) - P \left(\hat{R}^\kappa_{\{(x,\omega)\}},\phi,0 \right)\right\|,
\end{equation}
where $\{(x,\omega)\}$ is the counterexample pair computed in Section \ref{sec:verification}.

In reporting the synthesized controllers, its parameters are rounded from six to three significant numbers for space considerations.

\rev{Finally, in this section we denote the logic function indicating set membership of a set $Y$ by $\varphi_Y$, i.e. given a set $Y$ in the form
\begin{equation*}
Y := \left\{ x \in \R^n ~\middle| ~ \bigvee\nolimits_{i} \bigwedge\nolimits_{j} h_{ij}(x) \sim 0\right\}, ~ \sim \in \{\geq, > \},
\end{equation*}
where $h_{ij}:\R^n \rightarrow \R$, we have $\varphi_Y= \bigvee_{i} \bigwedge_{j} h_{ij}(x) \sim 0$.
}

\subsection{Car benchmark}

\begin{figure*}[t]
     \subfloat[]{%
     \includegraphics[width=0.24\textwidth]{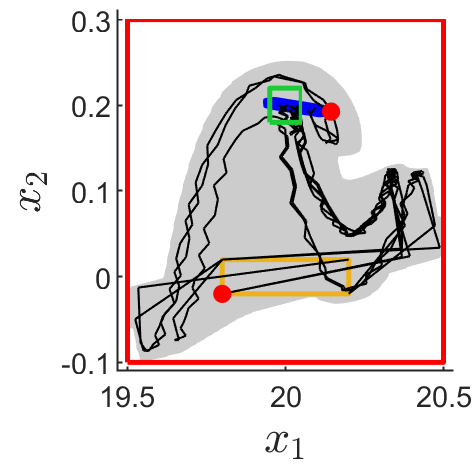}
     }
     \subfloat[]{%
     \includegraphics[width=0.24\textwidth]{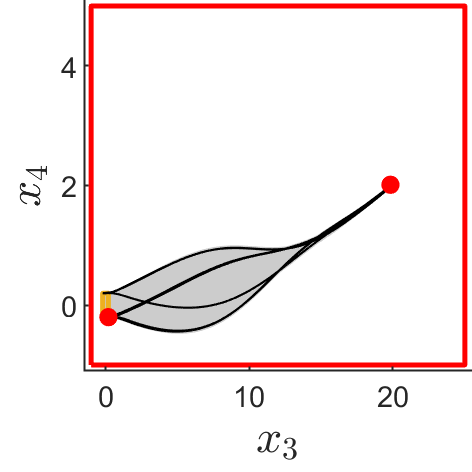}
     }
     \subfloat[]{%
     \includegraphics[width=0.24\textwidth]{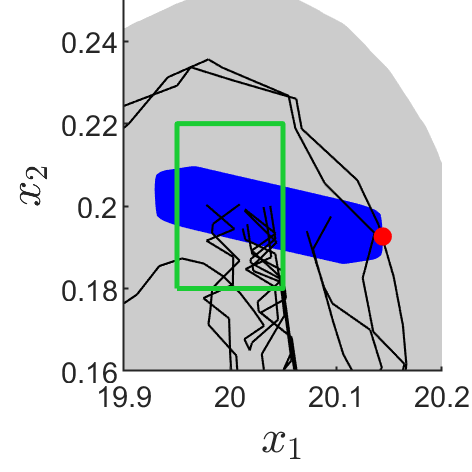}
     }
     \subfloat[]{%
     \includegraphics[width=0.24\textwidth]{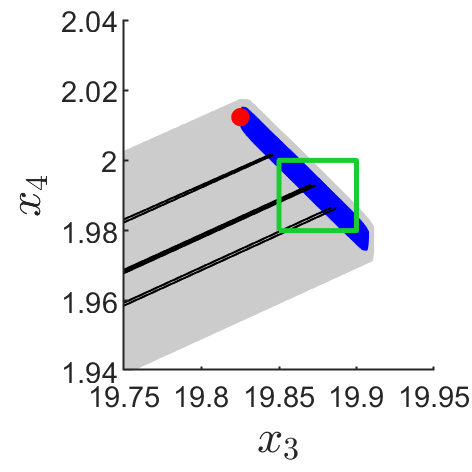}
     }
     \caption{Reachable set for the first controller for the car benchmark, which violates the desired controller specification. Figures (c) and (d) illustrate the reachable set near the goal set. Red dots: a point in the final reachable set that is outside of the goal set and its corresponding initial state, yellow: initial set, green: goal set $G$, gray: reachable set, red: safe set $S$, blue: reachable set at $t = 1$, black: example of simulation traces.}
     \label{fig:car_nosatV}      
   \end{figure*}

Let us consider a kinematic model of a car from \cite{Schuermann2017c}:
\begin{align*}
\left\{ \!  \begin{array}{ll}
f(x,u,\omega)  \! = \! ( u_1 \!+ \!w_1, u_2 \!+ \!w_1, x_1 \! \cos(x_2), x_1 \! \sin(x_2)^T,\\
I \! = \!  [19.9, 20.2] \times [-0.02,0.02] \times [-0.2,0.2]^2, \\
\Omega \! = \!   [-0.5, 0.5] \times [-0.02, 0.02].
\end{array} \right.
\end{align*}
where the states $x_1$, $x_2$, $x_3$, $x_4$ denote the velocity, orientation, and $x$ and $y$ position of the car, respectively. Furthermore, $u_1$ and $u_2$ denote the inputs and $w_1$ and $w_2$ disturbances. The sampling time $\eta$ of the sampled-data controller is set to be 0.025 seconds. 
Similarly to \cite{Schuermann2017c}, we consider a ``turn left" maneuver over a time interval $T= [0,1]$, where within $T$, the trajectories stay within the safe set $S$ and at the final time instant, the system is in the goal set, captured by the STL specification:
\begin{equation}
\varphi_1 = \always_{[0,1]} \varphi_S \wedge \always_{\{1\}} \varphi_G.
\end{equation}
We consider the following safe set $S$ and goal set $G$:
\begin{align*}
S &= [19.5, 20.5] \! \times \![-0.1, 0.3]\! \times\! [-1,25]\! \times\![-1,5], \\
G &= [19.95, 20.05] \!\times\! [0.18, 0.22] \!\times\! [19.85,19.9]\! \times\![1.98,2]. 
\end{align*}
To guide the synthesis, we impose the reference-tracking controller structure from Section \ref{sec:reftrack} and therefore we first design a feedforward signal and reference trajectory using GGGP. 
For $u_\mathrm{ff}$, $x_\mathrm{ref}$, we use polynomial expressions as a function of time $t$, for the feedback law $\kappa$ we restrict the search space to reference-tracking controllers which are linear in the tracking error and polynomial in time:
\begin{equation}
\kappa(x,t) = u_\mathrm{ff}(t) + K(t)(x - x_\mathrm{ref}(t)),
\label{eq:controllerform}
\end{equation}
and for $\omega^i$ we consider saturated polynomials in time. This is done using the grammar with starting trees:
\begin{align*}
\mathcal{S}_{u_\mathrm{ff}} &= (\rmms{pol}{t},\rmms{pol}{t})^T,~\mathcal{S}_{x_{\mathrm{ref},i}}= \rmms{pol}{t}, \\
\mathcal{S}_{\kappa} &= u_\mathrm{ff} + \begin{pmatrix}
\rmms{pol}{t}, \dots, \rmms{pol}{t} \\
\rmms{pol}{t}, \dots, \rmms{pol}{t}
\end{pmatrix} (x-x_\mathrm{ref}), \\
\mathcal{S}_{\omega^i} &= \left(\mathrm{sat}_{(\underline{\omega}_1, \overline{\omega}_1)} ( \rmms{pol}{t}), \mathrm{sat}_{(\underline{\omega}_2, \overline{\omega}_2)} ( \rmms{pol}{t})\right)^T.
\end{align*}
Here, $\mathrm{sat}_{(\underline{\omega}_i, \overline{\omega}_i)} $ denotes a saturation function such that $\omega^i(t) \in \Omega$, where
$\mathrm{sat}_{(\underline{\omega}_i, \overline{\omega}_i)}(x) = \max(\underline{\omega}_i, \min(x, \overline{\omega}_i)).$
Finally, for each disturbance realization, we co-evolve the error signal $\varepsilon^i$ in \eqref{eq:errorsignal}, which is dependent on the candidate controller $\kappa$ and disturbance realization $\omega^i$:
\begin{align*}
\mathcal{S}_{\varepsilon^i} &= \delta \sigma f(t,x,\kappa(x),\omega^i),\\
\sigma &= \mathrm{diag}(\mathrm{sat}_{(-1,1)}(\rmms{pol}{t}), \dots, \mathrm{sat}_{(-1,1)}(\rmms{pol}{t})),
\end{align*}%
where $\mathrm{diag}$ denotes a diagonal matrix.
For the simulations and reachability analysis, we use a sampling time of $0.025$ seconds and $0.0125$ seconds, respectively.

First, a feedforward control input and reference trajectory for a nominal initial condition are synthesized as described in Section \ref{sec:reftrack}. An example of a found feedforward controller and corresponding reference trajectory are shown in Table \ref{tab:controllers}.
For 10 independent runs, the average synthesis time of $u_\mathrm{ff}$ and the reference trajectory per dimension $x_{\mathrm{ref},i}$ is shown in Table \ref{tab:results}. Using these $u_\mathrm{ff}$ and $x_\mathrm{ref}$ as building blocks for the controller, $\kappa$ is synthesized as described in step A1). An example of a synthesized $K(t)$ in \eqref{eq:controllerform} is given by
\begin{align*}
K(t) = \begin{pmatrix}
-41.5 &	-6.48 t^2	&-84.3958&	9.45 \\
3.58&	-30.1&	-8.22&	3.62 t1-49.2 t^2
\end{pmatrix}.
\end{align*}
The corresponding reachable set is shown in Figure \ref{fig:car_nosatV}. We observe that the final reachable set is not within the goal set. The red dots represent the violation and the corresponding initial condition. After refining the controller iteratively, an example of a controller satisfying $\varphi_1$ after 3 refinements is shown in Table \ref{tab:controllers}.

For 10 independent synthesis runs of $\kappa$, statistics on the number of generations, number of refinements, complexity in terms of number of non-terminals, and computation time is shown in Tables \ref{tab:results} and \ref{tab:results}and Figure \ref{fig:refincomparison}. In most cases, a solution was obtained around 3 refinements. However, due to the stochastic nature of the approach, in one case it took 20 refinements before a solution was found.

\subsection{Input saturation}
\label{sec:STLinput}

In our general framework, we do not canonically consider input saturation. Input saturation can be considered in multiple ways, such as restricting the grammar of the controller to include a saturation function, or even a continuous approximation using e.g. a sigmoid function. However, the downside of such an approach is that the reachability analysis under these functions is typically challenging for state-of-the-art reachability tools, due to the strong nonlinearity or hybrid nature. Instead, for illustrative purposes, we incorporate the constraint within the STL specification, such that for all states in the reachable set the saturation bounds are not exceeded. Let us revisit the car benchmark, where we consider the same input constraints as in \cite{Schuermann2017c}, namely $u \in \overline{U} = [-9.81,9.81]\times [-0.4, 0.4]$. The STL specification is extended to:
\begin{equation}
\varphi_2 = \varphi_1 \wedge \always_{[0,1]} \varphi_U
\end{equation}
with 
\begin{equation}
U = \left\{x \in \R^n ~|~ \kappa(x) \in \overline{U} \right\}.
\label{eq:phiu}
\end{equation}

The synthesis statistics are shown in Tables \ref{tab:results} and \ref{tab:results} and Figure \ref{fig:refincomparison}. An example of a synthesized $K(t)$ in  \eqref{eq:controllerform} is given by
\begin{equation*}
K(t) = \begin{pmatrix}
-18.1+18.2 t-65.9 t^6	& 0.22 t	 \\
0	& -8.26-41.8 t \\
-29.6-48.7 t & 0 \\
-11.2 t &	-33.1 t^2
\end{pmatrix}^T.
\end{equation*}
In most cases, a solution was found in around 4 to 5 refinements, with the exceptions of two runs with 20 and 40 refinements, respectively.

\subsection{Path planning for simple robot}

\begin{figure}
\centering%
 \subfloat[]{%
     \includegraphics[width =0.3\textwidth]{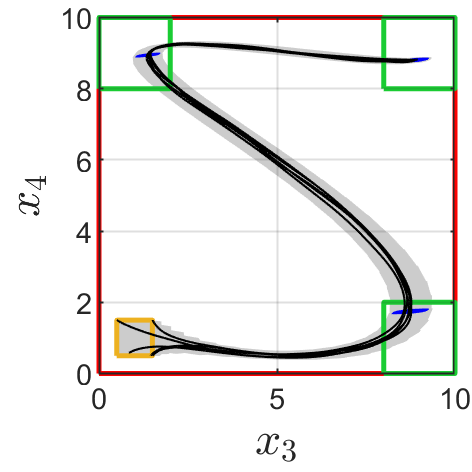}
     }
     \subfloat[]{%
     \includegraphics[width =0.3\textwidth]{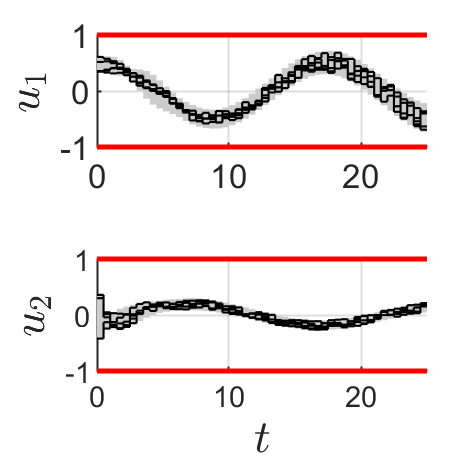}
     }
\caption{Reachable set of a found controller for the path planning benchmark. (a) Reachable set of the $x$-$y$ position. (b) Reachable set of the input over time. Yellow: initial set, gray: reachable set, red: safe set $S$ and input constraints, green: target sets $P_1$, $P_2$, and $P_3$, black: selection of simulated trajectories, blue: reachable sets at certain time instances within one of the target sets.}
\label{fig:reachpp}
\end{figure}

Let us consider the path-planning problem for a simple robot adopted from \cite{Lindemann2017}. We deviate from \cite{Lindemann2017} in considering the system in continuous time and consider bounded disturbances. The system is described by:
\begin{align*}
\left\{ \begin{array}{ll} f(x,u,\omega) = (u_1+w_1, u_2+ w_2, x_1, x_2)^T, \\
I = \{0\}^2 \times [0.5,1.5]^2,~
\Omega = [-0.05,0.05]^ 2,
\end{array} \right.
\end{align*}
where the state vector represents the $x$-velocity, $y$-velocity, $x$-position and $y$-position, respectively. The sampling time of the sampled-data controller $\eta$ is set to be 0.5 seconds. 
Similar to \cite{Lindemann2017}, we consider the specification in which the system needs to remain in a safe set $S$ and eventually visit regions $P_1$, $P_2$ and $P_3$:
\begin{equation}
\varphi'  = \always_{[0,25]} \phi_S \wedge \event_{[5,25]} \phi_{P_1} \wedge \event_{[5,25]} \phi_{P_2} \wedge \event_{[5,25]} \phi_{P_3},
\end{equation}
with $S = \{ x\in \R^n \mid (x_3, x_4) \in [0,10]^2\}$, $P_1 = \{ x \in \R^n \mid (x_3, x_4) \in [8,10]^2\}$, $P_2 = \{ x \in \R^n \mid (x_3, x_4) \in [8,10]\times [0,2]\} $, $P_3 = \{ x  \in \R^n\mid (x_3, x_4) \in [0,2] \times [8,10]\}$. In \cite{Lindemann2017}, the input is constrained s.t. $u \in \overline{U} = [-1,1]^ 2$. Similar to Section \ref{sec:STLinput}, we impose this constraint through the STL specification, yielding the following STL specification:
\begin{equation}
\varphi = \varphi'  \wedge \always_{[0,25]}\varphi_U,
\end{equation}
where $U$ is given by \eqref{eq:phiu}. We consider the same controller structure and grammar as the previous benchmark, with the exception of the grammar of the feedfoward input and reference trajectory. For these elements, we extend the grammar to expressions which can include trigonometric functions, by using the grammar in Table \ref{tab:PRexamp} and the starting trees $\mathcal{S}_{u_\mathrm{ff}} = (\rmms{expr}{t},\rmms{expr}{t})$ and $\mathcal{S}_{x_{\mathrm{ref},i}}= \rmms{expr}{t}$.
For the simulations and reachability analysis, we use a sampling time of $0.5$ seconds. The statistics on the synthesis is again shown in Tables \ref{tab:results} and \ref{tab:results} and Figure \ref{fig:refincomparison}.
An example of the controller elements $u_\mathrm{ff}$, $x_\mathrm{ref}$ and $K(t)$ of a synthesized controller are shown in Table \ref{tab:controllers}. The corresponding reachable set of the state and input is shown in Figure \ref{fig:reachpp}. Across 10 independent runs, commonly in 1 to 2 refinements a solution was found, with one run requiring 8 refinement.

\subsection{Landing maneuver}

\begin{figure}
\centering
\includegraphics[width = 0.6\linewidth]{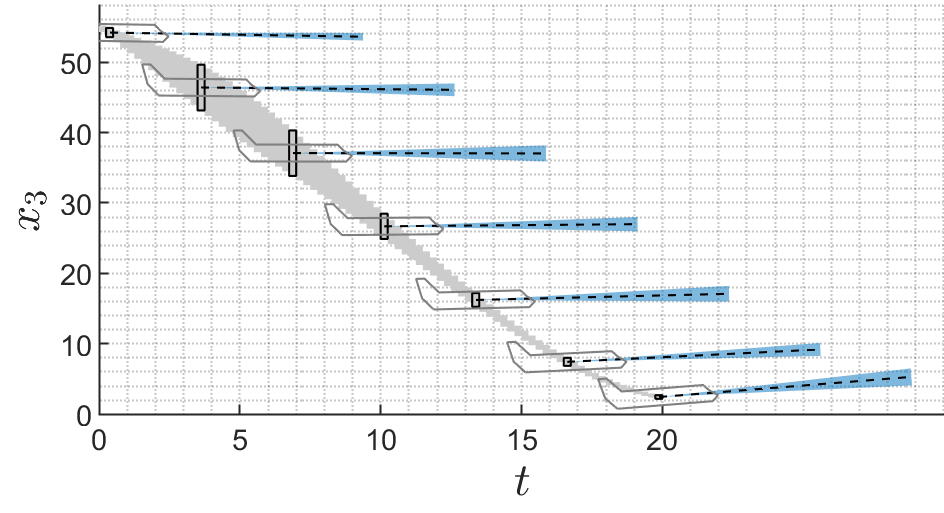}
\caption{Time evolution of the reachable set of the altitude $x_3$ under a synthesized controller for the landing maneuver. Gray: Reachable set over time of the altitude $x_3$. Blue: the set of the aircraft pitch $x_2 +u_2$ for 7 time intervals.}
\label{fig:aircraft}
\end{figure}

Let us consider the landing aircraft maneuver, adopted from \cite{Reissig2017}. The system model is given by
\begin{equation*}
\left\{ \!  \begin{array}{l}
f(x,\nu,\omega) \! =\!\! \begin{pmatrix}
\! \frac{1}{m}(\nu_1 \! \cos \nu_2 \! -\! D(\nu_2, x_1) - mg \! \sin x_2) \\
\! \frac{1}{m x_1} ( \nu_1  \!\sin \nu_2\! + \!  L(\nu_2,x_1) -mg \! \cos x_2)\!\! \\
x_1 \sin x_2
\end{pmatrix} \! \!,\\
D(\nu_2,x_1) = (2.7 +3.08 (1.15 +4.2 \nu_2)^2)x_1^2,\\
L(\nu_2, x_1) = (68.6 (1.25 + 4.2 \nu_2)) x_1^2,\\
\nu_i = u_i + w_i,~~i = 1,2,\\
I = [80,82] \times [-2^\circ,-1^\circ] \times \{55\} \\
\Omega = [-5 \cdot 10^3 ,-5 \cdot 10^3 ] \times [-0.25^\circ , 0.25^\circ],
\end{array}
\right.
\end{equation*}
where the states $x_1$, $x_2$, $x_3$ denote the velocity, flight path angle and the altitude of the aircraft, $\nu_i$ denotes a disturbed input, where $u_1$ denotes the thrust of the engines and $u_2$ the angle of attack. Finally, $D(\nu,x_1)$ and $L(\nu,x_1)$ denote the lift and drag, respectively, and $m = 60 \cdot 10^3$ kg, $g = 9.81 \mathrm{m}/\mathrm{s}^2$. The sampling time of the sampled-data controller is set at $\eta =0.25$ seconds. Compared to \cite{Reissig2017}, we do not consider measurement errors, but the proposed framework can be adapted arbitrarily to accommodate this type of disturbance. We define the following safe set, goal set and input bounds:
\begin{align*}
S = &[58, 83] \times [-3^\circ, 0^\circ] \times [0,56],\\
G =& [63, 75] \times ([-2^\circ, -1^\circ] \times [0,2.5]) \\
& \cap \{ x\in \R^3 \mid x_1 \sin x_2 \geq -0.91 \}, \\
\overline{U} =& [0, 160 \cdot 10^3] \times [0^\circ,10^\circ],
\end{align*}
and consider the following specification:
\begin{equation}
\varphi = (\varphi_S \wedge \varphi_U) \mathcal{U}_{[18,20]} \varphi_G,
\end{equation}
where the set $U$ is given by \eqref{eq:phiu}. That is, trajectories are always within the safe set and satisfy the input constraints, until between 18 and 20 seconds the goal set is reached. 

We use the same controller structure and grammar as the path-planning problem. For the simulations and reachability analysis, we use a sampling time of $0.25$ seconds. The algorithm settings are shown in Table \ref{tab:settings}. The statistics of 10 independent synthesis runs are again shown in Tables \ref{tab:results} and \ref{tab:results} and Figure \ref{fig:refincomparison}. An example of the controller elements $u_\mathrm{ff}$, $x_\mathrm{ref}$ and $K(t)$ of a synthesized controller are shown in Table \ref{tab:controllers}. The corresponding reachable set of the altitude over time, as well as the reachable sets of the pitch angles at multiple time instances are shown in Figure \ref{fig:aircraft}.

\rev{
\subsection{Scalability: platoon}
Consider a platooning system \cite{schurmann2017optimal}, described by:
\begin{align*}
\left\{ \begin{array}{ll} 
f(x,u,w) = (f_1(x,u,w), \dots, f_N(x,u,w))^T, \\
f_i(x,u,w) = ( x_{2i},  \delta u_i + \delta w_i), ~ i = \{1, \dots, N\}, \\
\delta u_1 = u_1,~\delta w_1 = w_1,  \\
\delta u_i = u_{i-1}-u_i,~ w_i = w_{i-1} - w_i, ~i = \{2, \dots N\}, \\
I = [-0.2,0.2]\times[19.8,20.8]\\
\hspace{0.7cm} \times ([0.8,1.2] \times[-0.2,0.2])^{N-1}, \\
\Omega =[-1,1]^N,
\end{array} \right.
\end{align*}

where $N$ denotes the number of vehicles, $x_1$, $x_2$ the position and velocity of the first vehicle, $x_{2i}$, $x_{2i+1}$ the relative position and relative velocities of vehicle $i$, and the input $u_i$ denotes the acceleration of vehicle $i$. We consider a sampling time $\eta$ of 0.05 seconds. The specification involves the acceleration of the platoon up to a goal set within a second, subjected to input constraints, while each vehicle maintains a safe distance, which is captured by the STL formula
\begin{equation}
    \varphi = \always_{[0,1]} \left(\varphi_\mathrm{s} \wedge \varphi_{\overline{U}} \right) \wedge \always_{\{1\}} \varphi_G, 
\end{equation}
where $\varphi_\mathrm{s} = \bigwedge_{i=2}^N x_{2i-1} \geq 0$, $ G = [20.8,21.2]\times[21.5,22.5]\times ([0.8,1.2] \times[-0.2,0.2])^{N-1}$, and $\overline{U} = [-10,10]^N$. 
We use the feed forward signal $u_{ff,i} =  1.4$ for all $i \in \{1,\dots,N\}$, where we exploit the fact that the desired control input for the agents after the first one is equal to the first one. 
The reference trajectory for this feedforward controller is given by:
\begin{equation*}
    x_\mathrm{ref}(t) = ( 20.3 t + 0.7 t^2, 20.3+1.4t, 1, 0, \dots, 1, 0)^T.
\end{equation*}
We impose the structure $\kappa_i(x,t) = \kappa_{i-1}(x,t)-\kappa_{i}'(x,t)$ for $i = 2, \dots, N$, where for $\kappa_i(x,t)$ and $\kappa_i'(x,t)$ we use the same control the same controller structure and grammar as the path-planning problem. For the simulations and reachability analysis, we use a sampling time of $0.05$ seconds. The algorithm settings are shown in Table \ref{tab:settings}. 

Given a maximum of 5000 GGGP generations, for $N = 2$, in 10 independent runs a controller was found, for $N =3$, in 9 out of 10 runs, and for $N =4$, no solutions were found. The results statistics for the successful synthesis runs are again again shown in Tables \ref{tab:results} and \ref{tab:results} and Figure \ref{fig:refincomparison}.
}

\rev{%
\subsection{Discovering structures from scratch: spacecraft}
\label{sec:discovering}
Let us consider a simplified model of spacecraft \cite{Brockett1983}, described by:
\begin{align*}
\left\{ \begin{array}{l}
    f(x,u,w) = \begin{pmatrix}
    u_1 +w_1, & u_2+w_2 & x_1 x_2 
    \end{pmatrix}^T,\\
    I = [-0.5,0.5]^2 \times [1,2],~ \Omega = [-0.1,0,1]^2, 
    \end{array} \right.
\end{align*}
where states denote the angular velocity and the inputs the control torques aligned with the principle axes. We consider a sampling time $\eta$ of $0.1$ seconds. For the simulations and reachability analysis, we use a sampling time of $0.1$ seconds.
Note that for stabilization, linearization methods are not appropriate, as the system linearized around points in the set $\{x \in \R^3 \mid x_1, x_2 = 0\}$ are not controllable. The goal is to control the system in finite time to a set around the origin $G = [-0.2,0.2]^3$ and the control input is constrained s.t. $u \in \overline{U} = [-5,5]^2$, which is captured by the STL specification 
\begin{equation*}
    \varphi = \always_{[0,5]}\phi_U \wedge \always_{\{5\}} \varphi_G.
\end{equation*}
For the disturbance, we use the same grammar as in the previous case studies. For the controller we consider polynomial state feedback controllers  This is done using the starting tree $\mathcal{S}_u = (\rmms{pol}{x}, \rmms{pol}{x})$. The algorithm settings are shown in Table \ref{tab:settings}. The results of 10 independent synthesis runs are again shown in Tables \ref{tab:results} and \ref{tab:results} and Figure \ref{fig:refincomparison}. An example of a found controller is given by
\begin{equation*}
    \kappa(x) = (-2.056 x_1-2.233 x_3, -2.034 x_2+2.071 x_3)^T.
\end{equation*}
Of the 10 synthesized controllers, 8 controllers have the same structure as the above controller.

}

\begin{table*}[t]
\caption{Statistics over an average of 10 independent synthesis runs. Total gen.: total number of GGGP generations for $\kappa$ before a solution was found; Total ref.: total number of refinements; Complexity: number of total non-terminals within the genotype of the synthesized controller.}
\label{tab:results1}
\centering
\smallskip
\scalebox{0.9}{
  \begin{tabular}{rcccccccccccccccccccccc}
  \toprule
  \multirow{2}{*}{\textbf{System}} & \multicolumn{3}{c}{\textbf{Total gen.}}& \multicolumn{3}{c}{\textbf{Total ref.}} & \multicolumn{3}{c}{\textbf{Complexity}} \Tstrut \\
    \cmidrule(lr){2-4}  \cmidrule(lr){5-7}  \cmidrule(lr){8-10} \\
& min & med & max & min & med & max & min & med & max  \Bstrut \\ \midrule
Car &63& 205.5& 1410& 3& 6& 19& 14& 27& 69\Tstrut\\
Constrained car& 84& 318& 933& 2& 5& 8& 24& 35.5& 56\\
Path planning &3& 16.5& 117& 1& 2.5& 9& 8& 11.5& 15 \\
Aircraft & 45& 342.5& 1165& 2& 5& 16& 24& 36& 58\Bstrut \\
%Platoon 2
Platoon N = 2 &  4& 64.5& 171& 1& 2& 7& 12& 21.5& 42\\
%Platoon 3
Platoon N = 3 & 522& 1611& 3210& 4& 5& 8& 30& 34& 56& \\
%Platoon
Spacecraft &  5& 67.5& 1350& 1& 2.5& 24& 14& 14& 22\\

 \bottomrule
  \end{tabular}%
  }
\end{table*}

\begin{table*}[t]
\caption{Statistics over an average of 10 independent synthesis runs. Time FF: average computation time of the feedforward components; Time: total time of the controller synthesis (excluding
the feedforward synthesis), GP $\kappa$: synthesis of candidate $\kappa$ using GGGP; GP $\omega$: disturbance realization optimization; RA: reachability analysis; CE: counterexample extraction; SMT: verifying the specification through an SMT solver; min: minimum; med: median; max: maximum. The average contribution percentages do not sum up to one, as the contribution of routines such as writing (SMT) files are not displayed.}
\label{tab:results}
\centering
\smallskip
\scalebox{0.9}{
  \begin{tabular}{rcccccccccccccccccccccc}
  \toprule
  \multirow{2}{*}{\textbf{System}} & \multicolumn{2}{c}{\textbf{Time FF [s]}}& \multicolumn{3}{c}{\textbf{Time [min]}} & \multicolumn{5}{c}{\textbf{Average contribution to total time $[\%]$}}  \Tstrut \\
   \cmidrule(lr){2-3}  \cmidrule(lr){4-6}  \cmidrule(lr){7-11}
& $u_\mathrm{ff}$ & $x_{\mathrm{ref},i}$& min & med & max & GP $\kappa$ & GP $\omega$ & RA & CE & SMT   \Bstrut \\ \midrule
Car &45.1 &1.2 & 16.5& 41.6& 204.1& 37.9& 26.2& 3.15& 19.3& 3.44\Tstrut\\
Constrained car & - & -&  28.0& 61.2& 117.0& 42.5& 17.2& 1.70& 15.8& 9.19\\
Path planning & 254.0 & 19.1 & 14.1& 23.8& 61.8& 7.61& 9.50& 3.05& 17.2& 27.8  \\
Aircraft &708.2  & 46.2 & 44.0& 165.1& 422.8& 36.7& 22.5& 12.9& 10.3& 7.71\Bstrut \\
%Platoon 2
Platoon N = 2&  - & - &  3.44&
9.30& 30.7& 29.0& 33.0& 3.42& 20.5& 1.48\\
%Platoon 3
Platoon N = 3 &  - & -&  67.6& \
207.9& 398.6& 60.7& 32.4 & 0.851& 3.12& \
0.403\\
%Platoon
Spacecraft &  - & - &   3.79& 
22.35 & 378.3& 25.0& 37.3& 2.15 & 18.13& 1.28\\

 \bottomrule
  \end{tabular}%
  }
\end{table*}

\begin{table*}[t]
\centering
\caption{Examples of synthesized controllers. Numerical values are rounded for space considerations.}
\label{tab:controllers}
\smallskip
\scalebox{0.6}{  
  \begin{tabular}{rccc}
  \toprule
  \textbf{System} & \textbf{Car (without input constraints)} & \textbf{Path planning} & \textbf{Aircraft}  \Tstrut \rule[-0.9ex]{0pt}{0pt} \\
   \midrule
%feedforward
$u_\mathrm{ff}$ & 
$\begin{pmatrix}
 0.01835 \\ 0.1995
 \end{pmatrix}$& 
 $\begin{pmatrix}
0.500 \cos (0.362 t+0.0733)\\
-0.190 \sin (0.678\, -0.324 t) 
\end{pmatrix}$ &  $\begin{pmatrix}
255.68 +107.57 t^2\\
0.00956 +0.00419 t
\end{pmatrix}$    \rule{0pt}{4.5ex}\\
  %reference
  $x_\mathrm{ref}$ & 
 $\begin{pmatrix}
 19.999 + 0.020567 t \\ 0.19954 t \\
  19.981 t - 0.10838 t^4\\ 1.9915 t^2.
  \end{pmatrix}$
  & $\begin{pmatrix}
0.03 t-3.81 \cos (0.361 t)+4.72\\
1.38 \sin (0.361 t)+0.024\\
0.406 t+1.88 \cos (0.312 t+0.949),\\
0.427 -0.583 \cos (0.765\, -0.325 t)
\end{pmatrix}$ & $\begin{pmatrix}
81.5 -0.380 t-1.28 \sin(0.393 +0.164 t) \\
(-0.164-1.59 \cdot 10^{-3} t) \cos(0.103 t)+0.138 \cos(0.120 t) \\
55.7 -0.674 \cos(0.354 t)-2.96 t \sin(0.788 +0.062 t)
\end{pmatrix}$ \rule{0pt}{7ex}\\
%Feedback gain
  $K(t)$
  &
  $\begin{pmatrix}
 -43.4	&3.94	&-89.6& 	307.3 t^2\\
 -8.28 t^5 &	-33.3&	-6.21&	-10.1
 \end{pmatrix}$
  &$\left(
\begin{array}{cccc}
 -0.264 & -0.125 t & 0 & 0.209 t \\
 0 & -0.204 & -0.781 & -1.35 \\
\end{array}
\right)$ & $\begin{pmatrix}
-2.67 t^3 &	-0.407-0.0636 t	& -0.788-0.461 t \\
-0.00607 &	-0.0217-0.237 t-0.0348 t^2	&-0.00023t^2
\end{pmatrix}$  \rule{0pt}{4.5ex} \rule[-2ex]{0pt}{0ex} \\
 \bottomrule
  \end{tabular}%
  }
\end{table*}
\section{Discussion}
\begin{figure*}[t]
   \subfloat[]{%
     \includegraphics[scale =0.8]{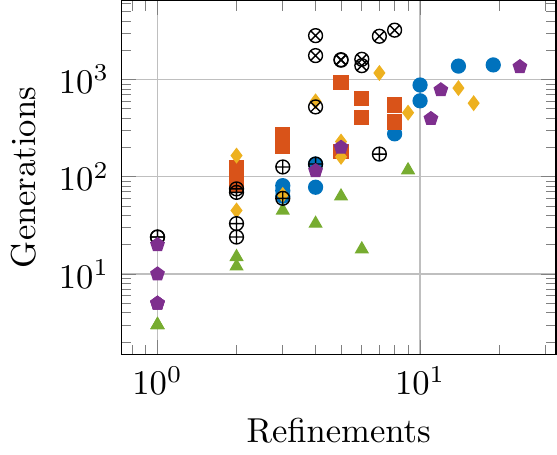}
     \label{fig:refvsgens}
     }
     \hfill
       \subfloat[]{%
     \includegraphics[scale =0.8]{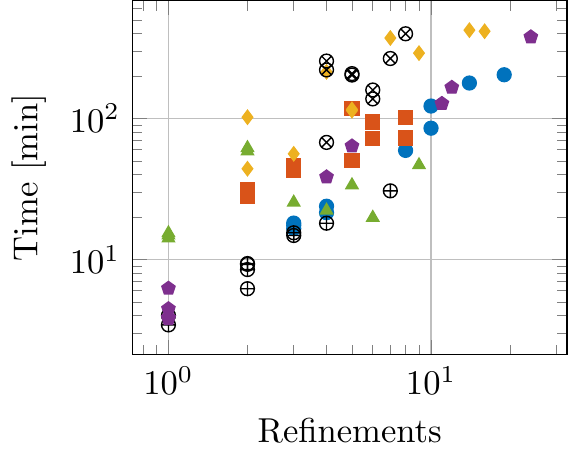}
     \label{fig:refsvstime}     
     }   
     \hfill
     \subfloat[]{%
     \includegraphics[scale =0.8]{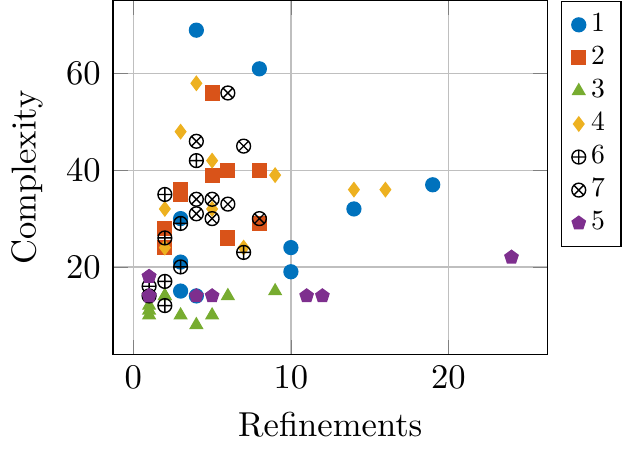}
     \label{fig:refsvscomp}
     }
     \caption{Number of refinements versus (a) number of GGGP generations, (b) time in minutes, and (c) complexity of the controller, measured in number of non-terminals, for systems 1) car, 2) constrained car, 3) path planning, 4) aircraft, 5) platoon $N =2$, 6) platoon $N=3$, 7) spacecraft.}
     \label{fig:refincomparison}      
\end{figure*}
We discuss now the results from Section \ref{sec:casestudies} and relate them to the results in the literature. 
Recall that a GGGP \textit{generation} is the cycle of creating a new population through fitness evaluation, selection and applying genetic operators. A \textit{refinement} is defined as the cycle of proposing a candidate solution based on GGGP, validation using reachability analysis, and extracting counterexamples. Therefore, in each refinement, there are one or multiple GGGP generations. First of all, Figure \ref{fig:refvsgens} shows a polynomial relation between the number of refinements and the total number of GGGP generations. Secondly, Figure \ref{fig:refsvstime} shows a polynomial relation between the number of refinements versus the total computation time. Finally, Figure \ref{fig:refsvscomp} illustrates that more refinements does not imply that complexity of the controller increases. However, the complexity of the found controller does seem to be dependent on the system and STL specification.

While the computation time is related to the number of refinements, this relationship depends on the STL specification and the dynamics. For the car benchmark without and with input constraints, we observe that the added constraints within the STL specification increased the required number of generations, and typically required more time per refinement. Hence, the total computation time heavily depends on the STL specification, as expected. Additionally, we observe an increase in the median of the complexity of the resulting controllers. 

\rev{With the platoon example, we see that for an increase in state dimension the number of generations required to find a solution significantly increases. This is expected, as the search space is significantly larger. For $N =4$, no solutions were found within 5000 GGGP generations. However, it is worth nothing that the optimal solution for $N = 4$ in \cite{schurmann2017optimal} very tightly satisfies specification. However, general conclusions regarding computation time and system order cannot be drawn.} For example, the input-constraint car and path planning benchmarks are both four-dimensional systems, where the STL specification of the latter is more involved. Regardless, the path-planning problem has a lower computation time and requires less generations and number of refinements, indicating a dependency between the computation time and the dynamics of the system, which is also as expected.

\rev{The resulting offline time complexity of our proposal is clearly higher than alternative methods like the one in \cite{Schuermann2017c}, with synthesis time for the car benchmark of just around 10 seconds, or in \cite{Reissig2017}, taking about 700 seconds for the aircraft benchmark controller synthesis. 
Both these alternative methods are faster offline, at the cost of potentially much larger controllers to be stored: a linear controller for each sampling time in \cite{Schuermann2017c}; an exponentially growing number of entries in a look-up table with increasing system dimension in discretization-based methods like that of \cite{Reissig2017}.
However, for the systems for which synthesis is successful in both ours and discretization-based methods, the size of the resulting controllers do not seem prohibitive. Nonetheless, a more clear advantage of our approach is the ability to constraint the controller structure so as to produce controllers that are easier to \emph{understand} by end-users than e.g. the look-up-tables (or BDDs) of discretization-based methods.}

\rev{An alternative to alleviate the memory footprint of controllers is the use of MPC approaches, such as \cite{Lindemann2017} for the path-planning problem. These solutions require additional online computational complexity compared to the quick evaluation that our controllers enable. As an example, the most complex of our aircraft controllers just requires on average $4.7 \cdot 10^{-7}$ seconds to compute the control actions (evaluated using $\mathtt{timeit}$ in MATLAB running on an Intel i7-8750H CPU).}

\rev{In summary, the tests performed indicate that this approach may be competitive with respect to competing alternatives whenever the application at hand requires both low memory and computational footprint, but more importantly whenever there is a need to impose specific controller structures to improve interpretability of the controller.}

\rev{As most automated synthesis methods for the type of problems we handle, the proposed framework is not a complete method. That is, the method is not guaranteed to find a solution in a finite number of iterations, regardless of its existence. Nevertheless, for the presented case studies, in 10 independent runs a solution was always found. Since the search space is navigated nondeterministically, we observed that the number of GGGP generations, number of refinements and computation time can vary significantly for each run.}

\rev{As it has been highlighted, the offline computation time of our current implementation is not competitive with those in other references. Note however, that the performance of our implementation has not been optimized for speed, being a mix of Matlab and Mathematica code, or parallelization of individuals in GGGP, which consume most of the computation, c.f. Table \ref{tab:results}.}
\rev{Additionally, limiting the fragment of STL, e.g., to $h(s)$ linear, the robustness degree computation can be considerably simplified. If the robustness measure is also upper bounded in a non-conservative manner, the usage of SMT solvers becomes redundant. This would significantly reduce the computation time for benchmarks such as the path-planning problem. 
Finally, input constraints, currently part of the STL formula, can also be captured using saturation functions in our grammar, simplifying the synthesis. However, as a caveat, discontinuous functions such as saturation functions significantly complicate the reachability analysis.}

\section{Conclusion}
We have proposed a framework for CEGIS-based correct-by-design controller synthesis for STL specifications based on reachability analysis and GGGP. The effectiveness has been demonstrated based on a selection of case studies. While the synthesis time is outmatched by methods solving similar problems, the proposed method results in a compact closed-form analytic controller which is provably correct when implemented in a sampled-data fashion. This enables the implementation in embedded hardware with limited memory and computation resources.

\appendix
\section{Proof of Theorem \ref{thm:soundness}}
\label{app:proofs}
Theorem \ref{thm:soundness} is proven by induction over the structure of the RTL formula $\phi$ and subformula $\psi$. This is only done for the first statement in Theorem \ref{thm:soundness}, as the second statement is logically equivalent to the first, i.e.:
\begin{align*}
P(R, \phi,t) \!>\! 0\! \Rightarrow (R,t)\! \models\! \phi \equiv (R,t)\! \not \models \! \phi \! \Rightarrow \! P(R, \phi,t) \! \leq \! 0,\\
(R,t) \! \models \!  \phi \!\Rightarrow \! P(R, \phi,t) \! \geq \! 0 \equiv P(R, \phi,t)  \!<\! 0 \! \Rightarrow \! (R,t) \not\models  \phi.
\end{align*}

\begin{itemize}[leftmargin=*]
\item \textbf{Case $\psi =\true$:} By definition $x \models \psi$ and $\varrho(x,\psi)>0$.
\item \textbf{Case $\psi = h(x) \geq 0$:} For this formula $\psi$, the quantitative semantics is given by $\varrho(x,\psi) = h(x)$. \textbf{(i)} If $\varrho(x,\psi) > 0$, then $h(x) >0$, thus from the semantics it follows that $x \models \psi$. \textbf{(ii)} If  $x \models \psi$, then from the semantics we have $h(x) \geq 0$, thus from the quantitative semantics it follows that $\varrho(x,\psi) \geq 0$. 
\item \textbf{Case $\psi = \neg \psi_1$:} For this formula $\psi$, the quantitative semantics is given by $\varrho(x,\neg \psi_1) = -\varrho(x,\psi_1)$. \textbf{(i)} If $\varrho(x, \neg \psi_1) > 0$, then $\varrho(x,\psi_1) < 0$. By the induction hypothesis, we get $x \not \models \psi_1$ and thus from the semantics it follows that $x\models \neg \psi_1$. \textbf{(ii)} If $x \models \neg \psi_1$, then from the semantics we have $x \not \models \psi_1$. By the induction hypothesis and the equivalence $\varrho(x,\psi) > 0 \Rightarrow x \models \psi \equiv x \not \models \psi
\Rightarrow \varrho(x,\psi) \leq 0$, we get $\varrho(x,\psi_1) \leq 0$, thus $\varrho(x, \neg \psi_1)\geq 0$. 
\item \textbf{Case  $\psi = \psi_1 \wedge \psi_2$:} For this formula $\psi$, the quantitative semantics is given by $\varrho(x,\psi_1 \wedge \psi_2) = \min(\varrho(x,\psi_1),$ $\varrho(x,\psi_2))$.
\textbf{(i)} If $\varrho(x, \psi_1 \wedge \psi_2)\! > 0$, then $ \varrho(x,\psi_1) > 0$ and $\varrho(x,\psi_2)>0$. By the induction hypothesis, we get $x \models \psi_1$ and $x \models \psi_2$, thus from the semantics it follows that $x \models \psi_1 \wedge \psi_2$. \textbf{(ii)} If $x \models \psi_1 \wedge \psi_2$, then from the semantics we have $x \models \psi_1$ and $x \models \psi_2$. By the induction hypothesis, we get $ \varrho(x,\psi_1) \geq 0$ and $\varrho(x,\psi_2)\geq 0$, thus $\varrho(x, \psi_1 \wedge \psi_2)\geq 0$.
\item \textbf{Case $\phi = \mathcal{A}\psi$:} For this formula $\phi$, the quantitative semantics is given by $P(R,\mathcal{A}\psi,t) = \inf_{x\in R(t)} \varrho(x,\psi)$.
\textbf{(i)} If $P(R,\mathcal{A}\psi,t) >0$, then $\forall x \in R(t): \varrho(x,\psi) > 0$. By the induction hypothesis, $\forall x \in R(t): x \models \psi$, thus from the semantics we have $(R,t) \models \mathcal{A} \psi$. \textbf{(ii)} If  $(R,t) \models \mathcal{A} \psi$, then from the semantics we have $\forall x \in R(t): x \models \psi$. By the induction hypothesis, we get $\forall x\in R(t): \varrho(x,\psi) \geq 0$, thus $P(R,\mathcal{A}\psi,t) \geq 0$.
\item \textbf{Case $\phi =\phi_1 \vee \phi_2$:} For this formula $\phi$, the quantitative semantics is given by $P(R,\phi_1 \vee \phi_2,t) = \max(P(R,\phi_1,t),P(R,\phi_2,t))$. 
\textbf{(i)} If $P(R,\phi_1 \vee \phi_2,t) > 0$, then $ P(R,\phi_1,t) > 0$ or $P(R,\phi_2,t)>0$. By the induction hypothesis, we get $(R,t) \models \phi_1$ or $(R,t) \models \phi_2$, thus from the semantics it follows that $(R,t) \models \phi_1 \vee \phi_2$. \textbf{(ii)} If $(R,t) \models \phi_1 \vee \phi_2$, then from the semantics we have $(R,t) \models \phi_1$ or $(R,t) \models \phi_2$. By the induction hypothesis, we get $P(R,\phi_1,t) \geq 0$ or $P(R,\phi_2,t)\geq 0$, thus $P(R,\phi_1 \vee \phi_2,t) \geq 0$.
\item \textbf{Case $\phi = \phi_1 \wedge \phi_2$:} For this formula $\phi$, the quantitative semantics is given by $P(R,\phi_1 \wedge \phi_2,t) = \min(P(R,\phi_1,t),P(R,\phi_2,t)$. 
\textbf{(i)} If $P(R,\phi_1 \wedge \phi_2,t) > 0$, then $P(R,\phi_1,t) > 0$ and $P(R,\phi_2,t)>0$. By the induction hypothesis, we get $(R,t) \models \phi_1$ and $(R,t) \models \phi_2$, thus from the semantics it follows that $(R,t) \models \phi_1 \wedge \phi_2$. \textbf{(ii)} If $(R,t) \models \phi_1 \wedge \phi_2$, then from the semantics we have $(R,t) \models \phi_1$ and $(R,t) \models \phi_2$. By the induction hypothesis, we get $P(R,\phi_1,t) \geq 0$ and $P(R,\phi_2,t)\geq 0$, thus $P(R,\phi_1 \wedge \phi_2,t) \geq 0$.
\item \textbf{Case $\phi = \next_a \phi_1$:} For this formula $\phi$, the quantitative semantics is given by $P(R,\next_a \phi_1,t) = P(R,\phi,t+a)$. \textbf{(i)} If $P(R,\next_a \phi_1,t) >0$, then $P(R,\phi_1,t+a)>0$. By the induction hypothesis, we get $(R,t+a) \models \phi_1$, thus from the semantics we have $(R,t) \models \next_a \phi_1$. \textbf{(ii)} If  $(R,t) \models \next_a \phi_1$, then from the semantics we have $(R,t+a) \models \phi_1$. By the induction hypothesis, we get $P(R,\next_a \phi_1,t) \geq 0$.
\qed
\end{itemize}

% use section* for acknowledgment
\section*{Acknowledgment}
The authors would like to thank Bastian Sch\"urmann for fruitful discussions on this work. 

\bibliographystyle{IEEEtran}
% argument is your BibTeX string definitions and bibliography database(s)
\bibliography{MyBib}             
                            
\end{document}